\documentclass{article}

\usepackage{arxiv}
\usepackage[utf8]{inputenc} 
\usepackage{authblk}
\usepackage[caption = false]{subfig}
\usepackage{graphicx}
\usepackage{amssymb}
\usepackage{amsmath}
\usepackage{amsthm}
\usepackage[noend]{algpseudocode}
\usepackage{algorithm}
\usepackage{titletoc}
\usepackage{titlesec}
\usepackage[utf8]{inputenc}
\usepackage{enumerate}

\makeatletter
\renewcommand{\Function}[2]{%
	\csname ALG@cmd@\ALG@L @Function\endcsname{#1}{#2}%
	\def\jayden@currentfunction{#1}%
}
\newcommand{\funclabel}[1]{%
	\@bsphack
	\protected@write\@auxout{}{%
		\string\newlabel{#1}{{\jayden@currentfunction}{\thepage}}%
	}%
	\@esphack
}
\makeatother

\newcommand{\onlay}{{\rm ONLAY}}

\newtheorem{thm}{Theorem}[section]
\newtheorem{lem}[thm]{Lemma}
\newtheorem{prop}[thm]{Proposition}

\newtheorem{defn}{Definition}[section]

\usepackage[english]{babel}

\newcommand{\dfnn}[2]{ \textbf{\emph{[#1]}} {#2}}

\renewcommand{\vec}[1]{\mathbf{#1}}
\newcommand{\eself}{\hookrightarrow^{s}}
\newcommand{\eref}{\hookrightarrow^{r}}
\newcommand{\eancestor}{\hookrightarrow^{a}} 
\newcommand{\eselfancestor}{\hookrightarrow^{sa}} 
\newcommand{\erefz}{\hookrightarrow}
\newcommand{\efork}{\pitchfork}

\newcommand{\hibefore}{\mapsto}
\newcommand{\hbefore}{\rightarrow}
\newcommand{\concur}{\parallel}

\makeatletter
\def\BState{\State\hskip-\ALG@thistlm}
\makeatother

\title{\onlay: Online Layering for scalable asynchronous BFT system}

\author{Quan Nguyen}
\author{Andre Cronje}

\affil{FANTOM Lab\\ FANTOM Foundation}

\begin{document}
\maketitle

\begin{abstract}
	
This paper presents a new framework, namely \emph{\onlay}, for scalable asynchronous distributed systems.
In this framework, we propose a consensus protocol $L_{\phi}$, which is based on the Lachesis protocol~\cite{lachesis01}.

 At the core of $L_{\phi}$ protocol, it introduces to use layering algorithm  to achieve practical Byzantine fault tolerance (pBFT) in leaderless asynchronous Directed Acyclic Graph (DAG). Further, we present new online layering algorithms for the evolutionary DAGs across the nodes. Our new protocol achieves determistic scalable consensus in asynchronous pBFT by using assigned layers and asynchronous partially ordered sets with logical time ordering instead of blockchains. The partial ordering produced by $L_{\phi}$ is flexible but consistent across the distributed system of nodes. 

We then present the formal model of our layering-based consensus. The model is generalized that can be applied to abstract asynchronous DAG-based distributed systems.
	 
\end{abstract}


\newpage
\pagenumbering{arabic} 
\tableofcontents 
\newpage
\section{Introduction}\label{ch:intro}

Following the widespead success over cryptocurrences, recent advances in blockchain and its successors have offered secure decentralized consistent transaction ledgers in numerous domains including financial, logistics as well as health care sectors. There have been extensive work \cite{algorand16, algorand17, sompolinsky2016spectre, PHANTOM08} addressing some known limitations, such as long consensus confirmation time and high power consumption, of blockchain-powered distributed ledgers.

In distributed database systems, \emph{Byzantine} fault tolerance (BFT)~\cite{Lamport82} addresses the reliability of the system when up to a certain number (e.g., one-third) of the participant nodes may be compromised. Consensus algorithms~\cite{bcbook15} ensures the integrity of transactions between participants over a distributed network~\cite{Lamport82} and is equivalent to the proof of BFT in distributed database systems~\cite{randomized03, paxos01}. 

For deterministic, completely asynchronous system, Byzantine consensus is not guaranted  with unbounded delays~\cite{flp}. Though achieving consensus is completely feasible for nondeterministic system. 
In practical Byzantine fault tolerance (pBFT), all nodes can successfully reach a consensus for a block in the presence of a Byzantine node \cite{Castro99}. Consensus in pBFT is reached once a created block is shared with other participants and the share information is further shared with others \cite{zyzzyva07, honey16}.

There have been extensive research in consensus algorithms. 
Proof of Work (PoW)~\cite{bitcoin08}, used in the original Nakamoto consensus protocol in Bitcoin, requires exhausive computational work from participants for block generation. Proof Of Stake (PoS)~\cite{ppcoin12,dpos14} uses participants' stakes for generating blocks.
With recent technological advances and innovations, consensus algorithms  \cite{algorand16, algorand17, sompolinsky2016spectre, PHANTOM08} have addressed to improve the consensus confirmation time and power consumption over blockchain-powered distributed ledges. These approaches utilize directed acyclic graphs (DAG)~\cite{dagcoin15, sompolinsky2016spectre, PHANTOM08, PARSEC18, conflux18} to facilitate consensus. 
Examples of DAG-based consensus algorithms include Tangle~\cite{tangle17}, Byteball~\cite{byteball16}, and Hashgraph~\cite{hashgraph16}.
Lachesis protocol~\cite{lachesis01} presents a general model of DAG-based consensus protocols.

\subsection{Motivation}

Lachesis protocols, as introduced in our previous paper~\cite{lachesis01}, is a set of consensus protocols that create a directed acyclic graph for distributed systems. 
We introduced a Lachesis consensus protocol, called $L_0$, which is a DAG-based asynchronous non-deterministic protocol that guarantees pBFT. $L_0$ generates each block asynchronously and uses the OPERA chain (DAG) for faster consensus by confirming how many nodes share the blocks.

In BFT systems, an synchronous approach utilizes a broadcast voting and asks each node to vote on the validity of each block.  
Instead, we aim for an asynchronous system where we leverage the concepts of distributed common knowledge and network broadcast to achieve a local view with high probability of being a consistent global view. 
Each node receives transactions from clients and then batches them into an event block. The new event block is then communicated with other nodes through asynchronous event transmission. During communication, nodes share their own blocks as well as the ones they received from other nodes. Consequently, this spreads all information through the network.
The process is asynchronous and thus it can  increase throughput near linearly as nodes enter the network.

Lachesis protocols $L_0$ \cite{lachesis01} and $L_1$ \cite{fantom18} proposed new approaches which are better than the previous DAG-based approaches. However, both protocols $L_0$ and $L_1$ have posed some limitations in their algorithms, which
can be made simpler and more reliable.



In this paper we are interested in a new consensus protocol that is more reliable to address pBFT in asynchronous scalable DAG. 
Specifically, we are investigating to use the notion of \emph{graph layering} and \emph{hierarchical graph} in graph theory to develop an intuitive model of consensus.

Let $G$=($V$,$E$) be a directed acyclic graph. In graph theory, a layering of $G$ is a topological numbering $\phi$ of $G$ that maps  each vertex $v$ $\in$ $V$ of $G$ to an integer $\phi(v)$ such  that $\phi(v)$ $\geq$ $\phi(u)$ + 1 for every directed edge $(u,v)$ $\in$ $E$. 
That is, a vertex $v$ is laid at layer $j$, if $\phi(v)$=$j$, and the $j$-th layer of $G$ is $V_j$= $\phi^{-1}(j)$. In other words, a layering $\phi$ of $G$ partitions of the set of vertices $V$ into a finite number $l$ of non-empty disjoint subsets (called layers) $V_1$,$V_2$,$\dots$, $V_l$, such that $V$ = $\cup_{i=1}^{l}{V_i}$, and for every edge ($u$,$v$) $\in$ $E$, $u$ $\in$ $V_i$, $v$ $\in$ $V_j$, $1$ $\leq$ $i < j \leq l$.
For a layering $\phi$, $H$=($V$,$E$,$\phi$) is called a hierarchical (layered or levelled) graph. An $l$-layer hierarchical graph can be represented as $H$=($V_1$,$V_2$,$\dots$,$V_l$;$E$).

Figure~\ref{fig:dag-ex}(a) shows an example of block DAG, which is a local snapshot of a node in a three-node network. Figure~\ref{fig:dag-ex}(b) depicts the result of layering applied on the DAG. In the figure, there are 14 layers. Vertices of the same layer are horizontally aligned. Every edge points from a vertex of a higher layer to a vertex of a lower layer. The hierarchical graph depicts a clear structure of the DAG, in which the dependency between blocks are shown uniformly from top to bottom.

\begin{figure}
	\centering
	\subfloat[OPERA chain (block DAG)]{\includegraphics[width=0.32\linewidth]{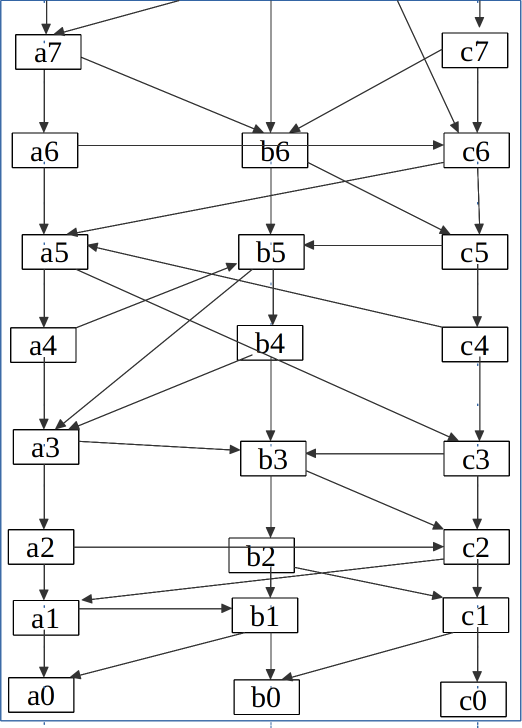}}
	\qquad\qquad
	\subfloat[H-OPERA chain (after layering)]{\includegraphics[width=0.3\linewidth]{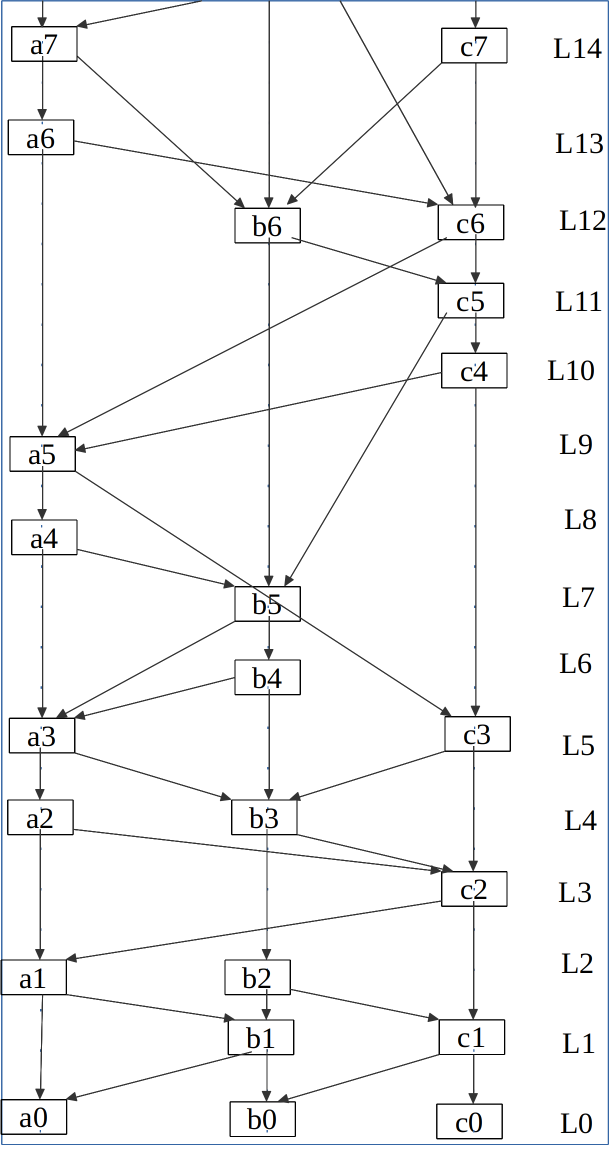}}
	\caption{An example of layered OPERA chain}
	\label{fig:dag-ex}
\end{figure}

The layering gives better structure of DAG. Thus, we aim to address the following questions:
1) Can we leverage the concepts of layering and hierarchical graph on the block DAGs?
2) Can we extend the layering algorithms to handle block DAGs that are evolving over time, i.e., on block creation and receiving? 
3) Can a model use layering to quickly compute and validate the global information and help quickly search for Byzantine nodes within the block DAG?
4) Is it possible to use hierarchical graph to reach consensus of the partical ordering of blocks at finality across the nodes?


\subsection{\onlay\ framework}

In this paper, we present \emph{\onlay}, a new framework for asynchronous distributed systems. \onlay\ leverages asynchronous event transmission for practical Byzantine fault tolerance (pBFT), similar to our previous Lachesis protocol~\cite{lachesis01}. 
The core idea of \onlay\ is to create a leaderless, scalable, asynchronous DAG.
By computing asynchronous partially ordered sets with logical time ordering instead of blockchains, \onlay\ offers a new practical alternative framework for distributed ledgers.

Specifically, we propose a consensus protocol $L_{\phi}$, which is based on Lachesis protocol~\cite{lachesis01}. $L_{\phi}$ protocol introduces an online layering algorithm  to achieve practical Byzantine fault tolerance (pBFT) in leaderless DAG. The protocol achieves determistic scalable consensus in asynchronous pBFT by using assigned layers and asynchronous partially ordered sets with logical time ordering instead of blockchains. The partial ordering produced by $L_{\phi}$ is flexible but consistent across the distributed system of nodes.  

We then present a formal model for the $L_{\phi}$ protocol that can be applied to abstract asynchronous DAG-based distributed system. The formal model is built upon the model of current common knowledge (CCK)~\cite{cck92}.

\begin{figure}[ht]
	\centering
	\includegraphics[width=\linewidth]{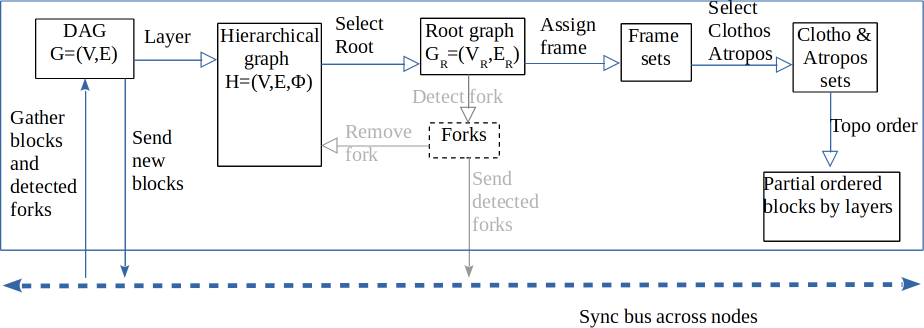}
	\caption{An overview of \onlay\ framework}
	\label{fig:onlay}
\end{figure}

Figure~\ref{fig:onlay} shows the overview of our \onlay\ framework.
In the framework, each node stores and maintain it own DAG. There are multiple steps, including layering, selecting roots, assigning frames, selecting Clothos and Atropos, and topologically ordering event blocks to determine consensus.
The major step of the framework compared to previous approaches is the layering step.

The main concepts of \onlay\ are given as follows:
\begin{description}
	\item[$\bullet$ Event block] An event block is a holder of a set of transactions created by a node and is then transported to other nodes. Event block includes signature, timestamp, transaction records and referencing hashes to previous (parent) blocks.
	\item[$\bullet$ $L_{\phi}$ protocol] sets the rules for event creation, communication and reaching consensus in \onlay.
	\item[$\bullet$ OPERA chain] is the local view of the DAG held by each node. This local view is used to determine consensus.
	\item[$\bullet$ Layering] Layering assigns every  block in OPERA chain a number so that every edge only points from high to low layers.
	\item[$\bullet$ L-OPERA chain] The L-OPERA chain is a hierarchical graph obtained from layering the DAG held by each node.
	\item[$\bullet$ Lamport timestamp] An event block $x$ is said \emph{Happened-before} $y$, if $y$ was created after $x$ was created. \emph{Lamport timestamp} is a timestamp assigned to each event block by a logical clock of a node. It is used to determine a partial order amongst the blocks.
	\item[$\bullet$ Root] An event block is called a \emph{root} if either (1) it is the first generated event block of a node, or (2) it can reach more than two-thirds of other roots.  A \emph{root set} $R_s$ contains all the roots of a frame. A \emph{frame}, denoted by $f$, is a natural number that is assigned to Root sets (and dependent event blocks). 
	\item[$\bullet$ Root graph] Root graph contains roots  as vertices and reachability between roots as edges.
	\item[$\bullet$ Clotho] A Clotho is a root satisfying that it is known by more than $2n/3$ nodes and more than $2n/3$ nodes know that information. 
	\item[$\bullet$ Atropos] An Atropos is a Clotho that is assigned with a consensus time.
	\item[$\bullet$ \onlay\ chain] \onlay\ chain is the main subset of the L-OPERA chain. It contains Atropos blocks and the subgraphs reachable from those Atropos blocks. 
\end{description}


\subsection{Contributions}

In summary, this paper makes the following contributions:
\begin{itemize}
\item We propose a new scalable framework, so-called \onlay, aiming for practical DAG-based distributed ledgers.
\item We introduce a novel consensus protocol $L_{\phi}$, which uses layer algorithm and root graphs, for faster root selection. $L_{\phi}$ protocol uses layer assignment on the DAG to achieve deterministic and thus more reliable consensus.
\item We define a formal model using continuous consistent cuts of a local view to achieve consensus via layer assignment.
\item We formalize our proofs that can be applied to any generic asynchronous DAG-based solution.
\end{itemize}

\subsection{Paper structure}

The rest of this paper is organised as follows.
Section~\ref{se:prelim} presents some preliminaries and background
for our \onlay\ framework. Section~\ref{se:consen} 
introduces our $L_{\phi}$ consensus protocol. The section describes our consensus algorithm that uses layering algorithm to achieve a more reliable and scalable solution to the consensus problem in BFT systems.
Section~\ref{se:onlay} presents important details of our \onlay\ framework.
Section~\ref{se:con} concludes.
Proof of Byzantine fault tolerance is described in Section~\ref{se:appendix}.

\section{Preliminaries}\label{se:prelim}

This section presents some background information as well as related terminologies used in \onlay\ framework.

\subsection{Basic Definitions}

The $L_{\phi}$ protocol sets rules for all nodes representing client machines that forms a network. 
In $L_{\phi}$, a (participant) node is a server (machine) of the distributed system.
Each node can create messages, send messages to, and receive messages from, other nodes. The communication between nodes is asynchronous.


\dfnn{Node}{Each machine participating in the protocol is called a \emph{node}. Let $n_i$ denote the node with the identifier of $i$. Let $n$ denote the total number of nodes.}
 


\dfnn{Process}{A process $p_i$ represents a machine or a \emph{node}. The process identifier of $p_i$ is $i$. A set $P$ = \{1,...,$n$\} denotes the set of process identifiers.}

\dfnn{Channel}{A process $i$ can send messages to process $j$ if there is a channel ($i$,$j$). Let $C$ $\subseteq$ \{($i$,$j$) s.t. $i,j \in P$\} denote the set of channels.}

The basic units of the protocol are called event blocks - a data structure created by a single node as a container to wrap and transport transaction records across the network. Each event block references previous event blocks that are known to the node. This makes a stream of DAG event blocks to form a sequence of history.

The history stored on each node can be represented by a directed acyclic graph $G$=($V$, $E$), where $V$ is a set of vertices and $E$ is a set of edges. Each vertex in a row (node) represents an event. Time flows bottom-to-top (or left-to-right) of the graph, so bottom (left) vertices represent earlier events in history. For a graph $G$, a path $p$ in $G$ is a sequence  of vertices ($v_1$, $v_2$, $\dots$, $v_k$) by following the edges in $E$.
Let $v_c$ be a vertex in $G$.
A vertex $v_p$ is the \emph{parent} of $v_c$ if there is an edge from $v_p$ to $v_c$.
A vertex $v_a$ is an \emph{ancestor} of $v_c$ if there is a path from $v_a$ to $v_c$.

\dfnn{Event block}{An event block is a holder of a set of transactions. The structure of an event block includes the signature, generation time, transaction history, and references to previous event blocks. The first event block of each node is called a \emph{leaf event}.}

Each node can create event blocks, send (receive) messages to (from) other nodes. 
Suppose a node $n_i$ creates an event $v_c$ after an event $v_s$ in $n_i$.  Each event block has exactly $k$ references. One of the references is self-reference, and the other $k$-1 references point to the top events of $n_i$'s $k$-1 peer nodes.

\dfnn{Top event}{An event $v$ is a top event of a node $n_i$ if there is no other event in $n_i$ referencing $v$.}



\dfnn{Ref}{An event $v_r$ is called ``ref" of event $v_c$ if the reference hash of $v_c$ points to the event $v_r$. Denoted by $v_c \eref v_r$. For simplicity, we can use $\erefz$ to denote a reference relationship (either $\eref$ or $\eself$).}

\dfnn{Self-ref}{An event $v_s$ is called ``self-ref" of event $v_c$, if the self-ref hash of $v_c$ points to the event $v_s$. Denoted by $v_c \eself v_s$.}

\dfnn{Self-ancestor}{An event block $v_a$ is self-ancestor of an event block $v_c$ if there is a sequence of events such that $v_c \eself v_1 \eself \dots \eself v_m \eself v_a $. Denoted by $v_c \eselfancestor v_a$.}

\dfnn{Ancestor}{An event block $v_a$ is an ancestor of an event block $v_c$ if there is a sequence of events such that $v_c \erefz v_1 \erefz \dots \erefz v_m \erefz v_a $. Denoted by $v_c \eancestor v_a$.}

For simplicity, we simply use $v_c \eancestor v_s$ to refer both ancestor and self-ancestor relationship, unless we need to distinguish the two cases.


\subsection{OPERA chain}

$L_{\phi}$ protocol uses the DAG-based structure, called the OPERA chain, which was introduced in Lachesis protocols~\cite{lachesis01}. 
For consensus, the algorithm examines whether an event block is reached by more than $2n/3$ nodes, where $n$ is the total number of participant nodes.

Let $G$=($V$,$E$) be a DAG. We extend with $G$=($V$,$E$,$\top$,$\bot$), where $\top$ is a pseudo vertex, called \emph{top}, which is the parent of all top event blocks, and 
$\bot$ is a pseudo vertex, called bottom, which is the child of all leaf event blocks.
With the pseudo vertices, we have $\bot$ happened-before all event blocks. Also all event blocks happened-before $\top$. That is, for all event $v_i$, $\bot \hbefore v_i$ and $v_i \hbefore \top$.

\dfnn{OPERA chain}{The OPERA chain is a graph structure stored on each node. The OPERA chain consists of event blocks and references between them as edges.}

The OPERA chain (DAG) is the local view of the DAG held by each node. 
OPERA chain is a DAG graph $G$=($V$,$E$) consisting of $V$ vertices and $E$ edges. Each vertex $v_i \in V$ is an event block. An edge ($v_i$,$v_j$) $\in E$ refers to a hashing reference from $v_i$ to $v_j$; that is, $v_i \erefz v_j$. This local view is used to identify Root, Clotho and Atropos vertices, and to determine topological ordering of the event blocks.

\dfnn{Leaf}{The first created event block of a node is called a leaf event block.}

\dfnn{Root}{An event block $v$ is a root if either (1) it is the leaf event block of a node, or (2) $v$ can reach more than $2n/3$ of the roots.}

\dfnn{Root set}{The set of all first event blocks (leaf events) of all nodes form the first root set $R_1$ ($|R_1|$ = $n$). The root set $R_k$ consists of all roots $r_i$ such that $r_i$ $\not \in $ $R_i$, $\forall$ $i$ = 1..($k$-1) and $r_i$ can reach more than 2n/3 other roots in the current frame, $i$ = 1..($k$-1).}

\dfnn{Frame}{Frame $f_i$ is a natural number that separates Root sets. The root set at frame $f_i$ is denoted by $R_i$.}

\dfnn{Creator}{If a node $n_i$ creates an event block $v$, then the creator of $v$, denoted by $cr(v)$, is $n_i$.}

\dfnn{Clotho}{A root $r_k$ in the frame $f_{a+3}$ can nominate a root $r_a$ as Clotho if more than 2n/3 roots in the frame $f_{a+1}$ dominate $r_a$ and $r_k$ dominates the roots in the frame $f_{a+1}$.}

\dfnn{Atropos}{An Atropos is a Clotho that is decided as final.}

Event blocks in the subgraph rooted at the Atropos are also final events. Atropos blocks form a Main-chain, which allows time consensus ordering and responses to attacks.

\subsection{Layering Definitions}\label{se:layering}

For a directed acyclic graph $G$=($V$,$E$), a layering is to assign a layer number to each vertex in $G$.

\dfnn{Layering}{A layering (or levelling) of $G$ is a topological numbering $\phi$ of $G$, $\phi: V \rightarrow Z$,  mapping the  set  of  vertices $V$ of $G$ to  integers  such  that $\phi(v)$ $\geq$ $\phi(u)$ + 1 for every directed edge ($u$, $v$) $\in E$.  If $\phi(v)$=$j$, then $v$ is a layer-$j$ vertex and $V_j= \phi^{-1}(j)$ is the jth layer of $G$.}

A layering $\phi$ of $G$ partitions the set of vertices $V$ into a finite number $l$ of \emph{non-empty} disjoint subsets (called layers) $V_1$,$V_2$,$\dots$, $V_l$, such that $V$ = $\cup_{i=1}^{l}{V_i}$. Each vertex is assigned to a layer $V_j$, where $1 \leq j \leq l$, such that every edge ($u$,$v$) $\in E$, $u \in V_i$, $v \in V_j$, $1 \leq i < j \leq l$.

\dfnn{Hierarchical graph}{
For a layering $\phi$, 
the produced graph $H$=($V$,$E$,$\phi$) is a \emph{hierarchical graph}, which is also called an $l$-layered directed graph and could be represented as
$H$=($V_1$,$V_2$,$\dots$,$V_l$;$E$).}

The spans of an edge $e$=($u$,$v$) with $u \in V_i$ and $v \in V_j$ is $j$-$i$ or $\phi(v)$-$\phi(u)$. 
If no edge in the hierarchical has a span greater than one then the hierarchical graph is \emph{proper}.

Let HG denote hierarchical graph.
Let $V^{+}(v)$ denote the outgoing neighbours of $v$; $V^{+}(v)$ =$\{u \in V | (v,u) \in E\}$. Let $V^{-}(v)$ denote the incoming neighbours of $v$; $V^{-}(v)$=$\{ u \in V | (u,v) \in E \}$.

\dfnn{Height}{The height of a hierarchical graph is the number of layers $l$.}

\dfnn{Width}{The width of a hierarchical graph is the number of vertices in the longest layer. that is, $max_{1 \leq i \leq l} |V_i|$}

There  are  three  common  criteria for a layering of a directed acyclic graph.
First, the hierarchical graph should be \emph{compact}. Compactness aims for minimizing the width and the height of the graph. Finding a layering by minimizing the height with respect to a given width is NP-hard.
Second, the  hierarchical  graph  should  be  proper, e.g.,  by  introducing dummy  vertices  into  the  layering  for  every  long  edge  ($u$,$v$) with $u \in V_i$ and $v \in V_j$ where $i <$ $j$-1. Each long edge ($u$,$v$), $u \in V_i$, $v \in V_j$ is replaced by a path ($u$,$v_1$,$v_2$,$\dots$,$v_l$,$v$) where a dummy vertex $v_p$, $i$+1 $\leq$ $p$ $\leq$ $j$-1, is inserted in each intermediate layer $V_{i+p}$. Third, the number of dummy vertices should be kept minimum so as to reduce the running time.

There are several approaches to DAG layering, which are described as follows.

\subsubsection{Minimizing The Height}

\dfnn{Longest Path layering}{
The longest path method  
layers vertex to layer $i$ where $i$ is the length of the longest path from a source.}

Longest path layering is a list scheduling algorithm produces hierarchical graph with the smallest possible height~\cite{bang2008digraphs,Sedgewick2011}.
The main idea is that given an acyclic graph we place the vertices on the $i$-th layer where $i$ is the length of the longest path to the vertex from a source vertex.

Algorithm \ref{algo:longestpath} shows the longest path layering algorithm. The algorithm picks vertices whose incoming edges only come from past layers ($Z$) and assigns them to the current layer $U$. When no more vertex satisfying that condition, the layer is incremented. The process starts over again until all vertices have been assigned a layer.

\begin{algorithm}[H]
	\caption{Longest Path Layering}\label{algo:longestpath}
	\begin{algorithmic}[1]
		\State Require:A DAG G=(V,E)
		\Function{LongestPathLayering}{}  
		\State $U \gets \emptyset$; $Z \gets \emptyset$; $l \gets 1$
		\While{$U \neq V$}
		\State Select vertex $v \in  V \setminus U $ with $V^{+}(v) \subseteq Z$
		\If{$v$ has been selected}
		\State $\phi(v) \gets l$
		\State $U \gets U \cup \{v\}$
		\EndIf
		\If{no vertex has been selected}
		\State $l \gets  l + 1$
		\State $Z \gets Z \cup U$
		\EndIf
		\EndWhile
		\EndFunction
	\end{algorithmic}
\end{algorithm}

The following algorithm similar to the above computes layerings of minimum height.
Firstly, all source vertices are placed in the first layer $V_1$.   
Then, the layer $\phi(v)$ for every remaining vertex $v$ is recursively defined by $\phi(v)$ = $max\{\phi(u) | (u, v) \in E\}$ + 1.  
This algorithm produces a layering where many vertices will stay close to the bottom, and hence the number of layers $k$ is kept minimized.   
The main drawback of this algorithm is that it may produce layers that are too wide. By using a topological ordering of the vertices~\cite{Mehlhorn84a}, the algorithm can be implemented in linear time $O$($|V|$+$|E|$).

\subsubsection{Fixed Width Layering}
The longest path layering algorithm minimizes the height, while  compactness of HG depends on both the width and the height. The problem of finding a layering with minimum height of general graphs is NP-complete if a fixed width is no less than three ~\cite{Garey1990}.
Minimizing the number of dummy vertices guarantees minimum height.

Now we present \emph{Coffman-Graham} (CG) algorithm, which considers a layering with a maximum width~\cite{Coffman1972}, 

\dfnn{Coffman-Graham algorithm}{Coffman-Graham is a layering algorithm in which a maximum width is given. Vertices are ordered by their distance from the source vertices of the graph, and are then assigned to the layers as close to the bottom as possible.}

Coffman-Graham layering algorithm was originated to solve multiprocessor scheduling to minimize the width (regardless of dummy vertices) as well as the height. The Coffman-Graham algorithm is currently the most commonly used layering method.

Algorithm~\ref{algo:coffmangraham} gives the Coffman-Graham algorithm. 
 The algorithm takes as input a reduced graph, i.e., no transitive edges are included in the graph, and a given width $w$. An edge $(u, v)$ is called transitive if a path ($u$=$v_1$,$v_2$,$\dots$,$v_k$=$v$) exists in the graph.
The Coffman-Graham algorithm works in two phases. First, it orders the vertices by their distance from the source vertices of the graph. In OPERA chain, the source vertices are the leaf event blocks. In the second phase, vertices are assigned to the layers as close to the bottom as possible. 

\begin{algorithm}[H]
	\caption{Coffman-Graham Layering}\label{algo:coffmangraham}
	\begin{algorithmic}[1]
		\State Require:A reduced DAG $G=(V,E)$
		\Function{CoffmanGraham}{W}
		\ForAll{$v \in V$}
		$\lambda(v) \gets \infty$
		\EndFor
		\For{$i$ $\gets$ 1 to $|V|$}
		\State Choose $v \in V$ with $\lambda(v)$= $\infty$ s.t. $|V^{-}(v)|$ is minimised
		\State $\lambda(v) \gets 1$
		\EndFor
		\State $l \gets 1; L_1 \gets \emptyset; U \gets \emptyset$
		
		\While{$U \neq V$}
		\State Choose $v \in V \setminus U$ such that $V^{+}(v) \subseteq U$ and $\lambda(v)$ is maximised
		
		\If{$|L_l| \leq W$ and $V^{+}(v) \subseteq L_1 \cup L2 \cup \dots \cup L_{l-1}$}
		\State $L_l \gets L_l \cup \{v\}$
		\Else
		\State	$l \gets l+ 1$
		\State  $L_l \gets \{v\}$
		\EndIf
		\State $U \gets U \cup \{v\}$
		\EndWhile
		\EndFunction
	\end{algorithmic}
\end{algorithm}

Lam and Sethi~\cite{Spinrad1985} showed that the number of layers $l$ of the computed layering with width $w$ is bounded by $l \leq (2-2/w).l_{opt}$, where $l_{opt}$ is the minimum height of all layerings with  width $w$. So, the Coffman-Graham algorithm is an exact algorithm for $w \leq 2$. In certain appplications, the notion of width does not consider dummy vertices.

\subsubsection{Minimizing the Total Edge Span}
The objective to minimize the total edge span (or edge length) is equivalent to minimizing the number of dummy vertices. It is a reasonable objective in certain applications such as hierarchical graph drawing. It can be shown that minimizing the number of dummy vertices guarantees minimum height. 
The problem of layering to minimize the number of dummy vertices can be modelled as an ILP, which has been proved to be very efficient experimentally, even though it does not guarantee a polynomial running time.

\subsection{Lamport timestamps}

Our Lachesis protocol $L_{\phi}$ relies on Lamport timestamps to define a topological ordering of event blocks in OPERA chain.  
By using Lamport timestamps, we do not rely on physical clocks to determine a partial ordering of events.

The ``happened before" relation, denoted by $\rightarrow$, gives a partial ordering of events from a distributed system of nodes. 
For a pair of event blocks $v_i$ and $v_j$, the relation "$\rightarrow$" satisfies: (1) If $v_i$ and $v_j$ are events of the same node $p_i$, and $v_i$ comes before $v_j$, then $v_i$ $\rightarrow$ $v_j$. (2) If $v_i$ is the send($m$) by one process and $v_j$ is the receive($m$) by another process, then $v_i$ $\rightarrow$ $v_j$. (3) If $v_i$ $\rightarrow$ $v_j$ and $v_j$ $\rightarrow$ $v_k$ then $v_i$ $\rightarrow$ $v_k$. 

\dfnn{Happened-Immediate-Before}{An event block $v_x$ is said Happened-Immediate-Before an event block $v_y$ if $v_x$ is a (self-) ref of $v_y$. Denoted by $v_x$ $\hibefore$ $v_y$.}

\dfnn{Happened-before}{An event block $v_x$ is said Happened-Before an event block $v_y$ if $v_x$ is a (self-) ancestor of $v_y$. Denoted by $v_x$ $\hbefore$ $v_y$.}

Happened-before relation is the transitive closure of happens-immediately-before.
``$v_x$ Happened-before $v_y$" means that the node creating $v_y$ knows event block $v_x$.  An event $v_x$ happened before an event $v_y$ if one of the followings happens: (a) $v_y \eself v_x$, (b) $v_y$ $\eref$ $v_x$,  or (c) $v_y$ $\eancestor$ $v_x$. The happened-before relation of events form an acyclic directed graph $G'$ = ($V$,$E'$) such that an edge ($v_i$,$v_j$) $\in$ $E'$ corresponding to an edge ($v_j$,$v_i$) $\in$ $E$.


\dfnn{Concurrent}{Two event blocks $v_x$ and $v_y$ are said concurrent if neither of them  happened before the other. Denoted by $v_x \concur v_y$.}

Given two vertices $v_x$ and $v_y$ both contained in two OPERA chains (DAGs) $G_1$ and $G_2$ on two nodes. We have the following:
(1) $v_x$ $\hbefore$ $v_y$ in $G_1$ if $v_x$ $\hbefore$ $v_y$ in $G_2$;  (2) $v_x$ $\concur$ $v_y$ in $G_1$ if $v_x$ $\concur$ $v_y$ in $G_2$.

\dfnn{Total ordering}{
	Let $\prec$ denote an arbitrary total ordering  of the nodes (processes) $p_i$ and $p_j$. 	\emph{Total ordering} is a relation $\Rightarrow$ satisfying the following: for any event $v_i$ in $p_i$ and any event $v_j$ in $p_j$, $v_i \Rightarrow v_j$ if and only if either (i) $C_i(v_i) < C_j(v_j)$ or (ii) $C_i(v_i)$=$C_j(v_j)$ and $p_i \prec p_j$.}

This defines a total ordering relation. The Clock Condition implies that if $v_i \rightarrow v_j$ then $v_i \Rightarrow v_j$. 

\subsection{State Definitions}

Each node has a local state, a collection of histories, messages, event blocks, and peer information, we describe the components of each.

\dfnn{State}{A (local) state of a process $i$ is denoted by $s_j^i$ consisting of a sequence of event blocks $s_j^i$=$v_0^i$, $v_1^i$, $\dots$, $v_j^i$.}

In a DAG-based protocol, each  event block $v_j^i$ is \emph{valid} only if the reference blocks exist before it. From a local state $s_j^i$, one can reconstruct a unique DAG. That is, the mapping from a local state  $s_j^i$ into a DAG is \emph{injective} or one-to-one. Thus, for \onlay, we can simply denote the $j$-th local state of a process $i$ by the DAG $g_j^i$ (often we simply use $G_i$ to denote the current local state of a process $i$).

\dfnn{Action}{An action is a function from one local state to another local state.}

Generally speaking, an action can be one of: a $send(m)$ action where $m$ is a message, a $receive(m)$ action, and an internal action. A message $m$ is a triple $\langle i,j,B \rangle$ where $i \in P$ is the sender of the message, $j \in P$ is the message recipient, and $B$ is the body of the message. 
In \onlay, $B$ consists of the content of an event block $v$. Let $M$ denote the set of messages.\\

Semantics-wise, there are two actions that can change a process's local state: creating a new event and receiving an event from another process.\\

\dfnn{Event}{An event is a tuple $\langle  s,\alpha,s' \rangle$ consisting of a state, an action, and a state. Sometimes, the event can be represented by the end state $s'$.}

The $j$-th event in history $h_i$ of process $i$ is $\langle  s_{j-1}^i,\alpha,s_j^i \rangle$, denoted by $v_j^i$.

\dfnn{Local history}{A local history $h_i$ of process $i$ is a (possibly infinite) sequence of alternating local states  --- beginning with a distinguished initial state. A set $H_i$ of possible local histories for each process $i$ in $P$.}

The state of a process can be obtained from its initial state and the sequence of actions or events that have occurred up to the current state. In the Lachesis protocol, we use append-only semantics. The local history may be equivalently described as either of the following: (1)
$h_i$ = $s_0^i$,$\alpha_1^i$,$\alpha_2^i$, $\alpha_3^i$ $\dots$, (2)
$h_i$ = $s_0^i$, $v_1^i$,$v_2^i$, $v_3^i$ $\dots$, (3)
$h_i$ = $s_0^i$, $s_1^i$, $s_2^i$, $s_3^i$, $\dots$.

In Lachesis, a local history is equivalently expressed as:
$h_i$ = $g_0^i$, $g_1^i$, $g_2^i$, $g_3^i$, $\dots$
where $g_j^i$ is the $j$-th local DAG (local state) of the process $i$.

\dfnn{Run}{Each asynchronous run is a vector of local histories. Denoted by
	$\sigma$ = $\langle h_1,h_2,h_3,...h_N \rangle$.}

Let $\Sigma$ denote the set of asynchronous runs. We can now use Lamport’s theory to talk about global states of an asynchronous system. A global state of run $\sigma$ is an $n$-vector of prefixes of local histories of $\sigma$, one prefix per process. The happens-before relation can be used to define a consistent global state, often termed a consistent cut, as follows.

\subsection{Consistent Cut}\label{sec:cck} 

An asynchronous system consists of the following sets: a set $P$ of process identifiers, a set $C$ of channels, a set $H_i$ of possible local histories for each process $i$, a set $A$ of asynchronous runs, a set $M$ of all messages.
Consistent cuts represent the concept of scalar time in distributed computation, it is possible to distinguish between a ``before'' and an ``after'', see CCK paper~\cite{cck92}.

\dfnn{Consistent cut}{A consistent cut of a run $\sigma$ is any global state such that if $v_x^i \rightarrow v_y^j$ and $v_y^j$ is in the global state, then $v_x^i$ is also in the global state. Denoted by $\vec{c}(\sigma)$.}

The concept of consistent cut formalizes such a global state of a run. A consistent cut consists of all consistent DAG chains. A received event block exists in the global state implies the existence of the original event block. Note that a consistent cut is simply a vector of local states; we will use the notation $\vec{c}(\sigma)[i]$ to indicate the local state of $i$ in cut $\vec{c}$ of run $\sigma$. 


The formal semantics of an asynchronous system is given via  the satisfaction relation $\vdash$. Intuitively $\vec{c}(\sigma) \vdash \phi$, ``$\vec{c}(\sigma)$ satisfies $\phi$,'' if fact $\phi$ is true in cut $\vec{c}$ of run $\sigma$.

We assume that we are given a function $\pi$ that assigns a truth value to each primitive proposition $p$. The truth of a primitive proposition $p$ in $\vec{c}(\sigma)$ is determined by $\pi$ and $\vec{c}$. This defines $\vec{c}(\sigma) \vdash p$.\\

\dfnn{Equivalent cuts}{Two cuts $\vec{c}(\sigma)$ and $\vec{c'}(\sigma')$ are equivalent  with respect to $i$ if: $$\vec{c}(\sigma) \sim_i \vec{c'}(\sigma') \Leftrightarrow \vec{c}(\sigma)[i] = \vec{c'}(\sigma')[i]$$}

\dfnn{$i$ knows $\phi$}{$K_i(\phi)$ represents the statement ``$\phi$ is true in all possible consistent global states that include $i$’s local state''. 
$$\vec{c}(\sigma) \vdash K_i(\phi) \Leftrightarrow \forall \vec{c'}(\sigma')   (\vec{c'}(\sigma') \sim_i \vec{c}(\sigma) \ \Rightarrow\ \vec{c'}(\sigma') \vdash \phi) $$}

\dfnn{$i$ partially knows $\phi$}{$P_i(\phi)$ represents the statement ``there is some consistent global state in this run that includes $i$’s local state, in which $\phi$ is true.''
$$\vec{c}(\sigma) \vdash P_i(\phi) \Leftrightarrow \exists \vec{c'}(\sigma) ( \vec{c'}(\sigma) \sim_i \vec{c}(\sigma) \ \wedge\ \vec{c'}(\sigma) \vdash \phi )$$}
		
\dfnn{Majority concurrently knows}{The next modal operator is written $M^C$ and stands for ``majority concurrently knows.''
The definition of $M^C(\phi)$ is as follows.

$$M^C(\phi) =_{def} \bigwedge_{i \in S} K_i P_i(\phi), $$ where $S \subseteq P$ and $|S| > 2n/3$.}

This is adapted from the ``everyone concurrently knows'' in CCK paper~\cite{cck92}.
In the presence of one-third of faulty nodes, the original operator ``everyone concurrently knows'' is sometimes not feasible.
Our modal operator $M^C(\phi)$ fits precisely the semantics for BFT systems, in which unreliable processes may exist.\\

\dfnn{Concurrent common knowledge}{The last modal operator is concurrent common knowledge (CCK), denoted by $C^C$. $C^C(\phi)$ is defined as a fixed point of $M^C(\phi \wedge X)$.}

CCK defines a state of process knowledge that implies that all processes are in that same state of knowledge, with respect to $\phi$, along some cut of the run. In other words, we want a state of knowledge $X$ satisfying: $X = M^C(\phi \wedge X)$.	
$C^C$ will be defined semantically as the weakest such fixed point, namely as the greatest fixed-point of $M^C(\phi \wedge X)$.It therefore satisfies:

$$C^C(\phi) \Leftrightarrow  M^C(\phi \wedge C^C(\phi))$$

Thus, $P_i(\phi)$ states that there is some cut in the same asynchronous run $\sigma$ including $i$’s local state, such that $\phi$ is true in that cut.\\

Note that $\phi$ implies $P_i(\phi)$. But it is not the case, in general, that $P_i(\phi)$ implies $\phi$ or even that $M^C(\phi)$ implies $\phi$. The truth of $M^C(\phi)$ is determined with respect to some cut $\vec{c}(\sigma)$. A process cannot distinguish which cut, of the perhaps many cuts that are in the run and consistent with its local state, satisfies $\phi$; it can only know the existence of such a cut.\\ 

\dfnn{Global fact}{Fact $\phi$ is valid in system $\Sigma$, denoted by $\Sigma \vdash \phi$, if $\phi$ is true in all cuts of all runs of $\Sigma$.
	$$\Sigma \vdash \phi 
	\Leftrightarrow (\forall \sigma \in \Sigma)(\forall\vec{c}) (\vec{c}(a) \vdash \phi)$$}
	
Fact $\phi$ is valid, denoted $\vdash \phi$, if $\phi$ is valid in all systems, i.e. 
	$(\forall \Sigma) (\Sigma \vdash \phi)$.\\

\dfnn{Local fact}{A fact $\phi$ is local to process $i$ in system $\Sigma$ if
	 $\Sigma \vdash (\phi \Rightarrow K_i \phi)$.}

\newpage
\section{$L_{\phi}$: Layering-based Consensus Protocol}\label{se:consen}
This section will then present the main concepts and algorithms of the $L_{\phi}$ protocol in our \onlay\ framework.
The key idea of our \onlay\ framework is to use layering algorithms on the OPERA chain. The assigned layers are then used to determine consensus of the event blocks.

For a DAG $G$=($V$,$E$), a layering $\phi$ of $G$ is a mapping $\phi: V \rightarrow Z$, such that $\phi(v)$ $\geq$ $\phi(u)$ + 1 for every directed edge ($u$,$v$) $\in E$.  If $\phi(v)$=$j$, then $v$ is a layer-$j$ vertex and $V_j$= $\phi^{-1}(j)$ is the $j$-th layer of $G$.  In terms of happened before relation, $v_i \hbefore v_j$ iff $\phi(v_i) < \phi(v_j)$.
Layering $\phi$ partitions the set of vertices $V$ into a finite number $l$ of \emph{non-empty} disjoint subsets (called layers) $V_1$,$V_2$,$\dots$, $V_l$, such that $V$ = $\cup_{i=1}^{l}{V_i}$. Each vertex is assigned to a layer $V_j$, where $1\leq j \leq l$, such that every edge ($u$,$v$) $\in E$, $u \in V_i$, $v \in V_j$, $1 \leq i < j \leq l$. Section~\ref{se:layering} gives more details about layering.

\subsection{H-OPERA chain}

Recall that OPERA chain, which is a DAG $G$=($V$,$E$) stored in each node.
We introduce a new notion of H-OPERA chain, which is built on top of the OPERA chain.
By applying a layering $\phi$ on the OPERA chain, one can obtain the hierarchical graph of $G$, which is called H-OPERA chain.

\begin{defn}[H-OPERA chain]
	An H-OPERA chain is the result hierarchical graph $H = (V,E, \phi)$.
\end{defn}

Thus, H-OPERA chain can also be represented by $H$=($V_1$,$V_2$,$\dots$,$V_l$;$E$).
In H-OPERA chain, vertices are assigned with layers such that each edge ($u$,$v$) $\in E$ flows from a higher layer to a lower layer i.e., $\phi(u)$ $>$ $\phi(v)$.

Note that, the hierarchical graph H-OPERA chain produced by \onlay\ is not proper. This is because there always exist a long edge ($u$,$v$) $\in E$ that spans multiple layers i.e. $\phi(u)$ - $\phi(v)$ $>$ 1.
Generally speaking, a layering can produce a proper hierarchical graph by introducing dummy vertices along every long edge.
However, for consensus protocol, such dummy vertices in the resulting proper hierarchical graph are also not ideal as they can increase the computational cost of the H-OPERA itself and of the following steps to determine consensus. Thus, in \onlay\ we consider layering algorithms that do not introduce dummy vertices.

A simple approach to compute H-OPERA chain is to consider OPERA chain as a static graph $G$ and then apply a layering algorithm such as either LongestPathLayer (LPL) algorithm or Coffman-Graham (CG) algorithm on the graph $G$.
LPL algorithm can achieve HG in linear time $O$($|V|$+$|E|$) with the minimum height. CG algorithm has a time complexity of $O$($|V|^2$), for a given $W$. The LPL and CG algorithms are given in  Section~\ref{se:layering}.

The above simple approach that takes a static $G$ works well when $G$ is small, but then may be not applicable when $G$ becomes sufficiently large. As time goes, the dynamic OPERA chain becomes larger and larger, and thus it takes more and more time to process.

\subsection{Online Layer Assignment}
To overcome the limitation of the above simple approach, we introduce online layering algorithms to apply on the evolunary OPERA chain to reduce the computational cost.

In \onlay, the OPERA chain in each node keeps evolving rather than being static. A node can introduce a new vertex and few edges when a new event block was created. The OPERA chain can also be updated with new vertices and edges that a node receives from the other nodes in the network.

\begin{defn}[Dynamic OPERA chain] 
	OPERA chain is a dynamic DAG which is evolving as new event block is created and received.	
\end{defn}

We consider a simple model of dynamic OPERA chain. Let DAG $G$=($V$,$E$) be the current OPERA chain of a node. Let $G'$=($V'$,$E'$)
denote the \emph{diff graph}, which consists of the changes to $G$ at a time, either at block creation or block arrival. We can assume that
the vertex sets $V$ and $V'$ are disjoint, similar to the edges $E$ and $E'$. That is,
 $V \cap V'$ = $\emptyset$ and $E \cap E'$ = $\emptyset$.
At each graph update, the updated OPERA chain becomes $G_{new}$=($V \cup V'$, $E \cup E'$).

An online version of layering algorithm takes as input a new change to OPERA chain and the current layering information and then efficiently computes the layers for  new vertices.

\begin{defn}[Online Layering]
Online layering is an incremental algorithm  to compute layering of a dynamic DAG $G$ efficiently as $G$ evolves.	
\end{defn}

Specifically, we propose two online layering algorithms. \emph{Online Longest Path Layering} (O-LPL) is an improved version of the LPL algorithm. \emph{Online Coffman-Graham} (O-CG) algorithm is a modified version of the original CG algorithm. The online layering algorithms assign layers to the vertices in diff graph $G'$=($V'$,$E'$) consisting of new self events and received unknown events.

\subsubsection{Online Longest Path Layering}

\begin{defn}[Online LPL] 
	Online Longest Path Layering (O-LPL) is a modified version of LPL to compute layering of a dynamic DAG $G$ efficiently.
\end{defn}	

Algorithm~\ref{algo:onlinelongestpath} gives our new algorithm to compute LPL layering for a dynamic OPERA chain. The O-LPL algorithm takes as input the following information: $U$ is the set of processed vertices, $Z$ is the set of already layered vertices, $l$ is the current height (maximum layer number), $V'$ is the set of new vertices, and $E'$ is the set of new edges.

\begin{algorithm}[H]
	\caption{Online Longest Path Layering}\label{algo:onlinelongestpath}
	\begin{algorithmic}[1]
		\State Require:A DAG $G$=($V$,$E$) and $G'$=($V'$,$E'$)
		\Function{OnlineLongestPathLayering}{$U$, $Z$, $l$, $V'$, $E'$}
		\State $V_{new}$ $\gets$ $V \cup V'$, $E_{new}$ $\gets$ $E \cup E'$
		\While{$U \neq V_{new}$}
		\State Select vertex $v \in  V_{new} \setminus U $ with $V_{new}^{+}(v) \subseteq Z$
		\If{$v$ has been selected}	
		\State $l' \gets max\{\phi(u) | u \in V^+_{new}(v) \}$ +1	
		\State $\phi(v)$ $\gets l'$ 
		\State $U \gets U \cup \{v\}$
		\EndIf
		\If{no vertex has been selected}		
		\State $l$ $\gets$  $l$ + 1
		\State $Z \gets Z \cup U$
		\EndIf
		\EndWhile
		\EndFunction
	\end{algorithmic}
\end{algorithm}

The algorithm assumes that the given value of width $W$ is sufficiently large. Section~\ref{sec:layerwidth} gives some discussions about choosing an appropriate value of $W$.
The algorithms has a worst case time complexity of $O$($|V'|$+ $|E'|$).

\subsubsection{Online Coffman-Graham Algorithm}

\begin{defn}[Online CG] 
	Online Coffman-Graham (O-CG) Algorithm is a modified version of the CG algorithm to compute layering of a dynamic DAG $G$ efficiently.
\end{defn}	

Algorithm~\ref{algo:onlinecoffmangraham} gives a new layering algorithm that is based on the original Coffman-Graham algorithm.
The algorithm takes as input the following: $W$ is the fixed maximum width, $U$ is the unprocessed vertices, $l$ is the current height (maximum layer), $\phi$ is the layering assignment, $V'$ is the set of new vertices, $E'$ is the set of new edges. The algorithms has a worst case time complexity of $O(|V'| ^2)$.

\begin{algorithm}[H]
	\caption{Online Coffman-Graham Layering}\label{algo:onlinecoffmangraham}
	\begin{algorithmic}[1]
		\State Require:A reduced DAG $G=(V,E)$ and $G'=(V',E')$
		\Function{OnlineCoffmanGraham}{$W$, $U$, $l$, $\phi$, $V'$, $E'$}
		\ForAll{$v \in V'$}
		$\lambda(v) \gets \infty$
		\EndFor
		\For{$i$ $\gets$ 1 to $|V'|$}
		\State Choose $v \in V'$ with $\lambda(v) = \infty$ s.t. $|V^{-}(v)|$ is minimised
		\State $\lambda(v) \gets 1$
		\EndFor
		\State $V_{new} \gets V \cup V'$; $E_{new} \gets E \cup E'$
		\While{$U \neq V_{new}$}
		\State Choose $v \in V_{new} \setminus U$ such that $V_{new}^{+}(v) \subseteq U$ and $\lambda(v)$ is maximised
		\State $l' \gets max\{\phi(u) | u \in V^+_{new}(v) \}$ +1			
		\If{$|\phi_{l'}| \leq W$ and $V_{new}^{+}(v) \subseteq \phi_1 \cup \phi_2 \cup \dots \cup \phi_{l'-1}$}
		\State $\phi_{l'} \gets \phi_{l'} \cup \{v\}$
		\Else
		\State	$l$ $\gets$ $l$+1
		\State  $\phi_l \gets \{v\}$
		\EndIf
		\State $U \gets U \cup \{v\}$
		\EndWhile
		\EndFunction
	\end{algorithmic}
\end{algorithm}

\subsection{Layer Width in BFT}\label{sec:layerwidth}
This section presents our formal analysis of the layering with respect to the Byzantine fault tolerance.

In distributed database systems, \emph{Byzantine} fault tolerance
addresses the functioning stability of the network in the presence of at most $n/3$ of Byzantine nodes, where $n$ is the number of nodes. A Byzantine node is a node that is compromised, mal-functioned, dies or adversarially targeted. Such nodes can also be called dishonest or faulty.

To address BFT of $L_\phi$ protocol, we present some formal analysis of the layering results of OPERA chain. 
For simplicity, let $\phi_{LP}$ denote the layering of $G$ obtained from LongestPathLayering algorithm. Let $\phi_{CG(W)}$ denote the layering of $G$ obtained from the Coffman-Graham algorithm with some fixed width $W$. 

\subsubsection{Byzantine free network}
We consider the case in which the network is free from any faulty nodes. In this setting, all nodes are honest and thus zero fork exists.

\begin{thm}
In the absense of dishonest nodes, the width of each layer $L_i$ is at most $n$.
\end{thm}

Since there are $n$ nodes, there are $n$ leaf vertices. We can prove the theorem by induction. The width of the first layer $L_1$ is $n$ at maximum, otherwise there exists a fork, which is not possible. We assume that the theorem holds for every layer from 1 to $i$. That is, the $L_i$ at each layer has the width of at most $n$. We will prove it holds for layer $i+1$ as well.
Since each honest node can create at most one event block on each layer, the width at each layer is at most $n$.
We can prove by contradiction. Suppose there exists a layer $|\phi_{i+1}| > n$. Then there must exist at least two events $v_p$, $v_q$ on layer $\phi_{i+1}$ such that $v_p$ and $v_q$ have the same creator, say node $n_i$. That means $v_p$ and $v_q$ are the forks of node $n_i$. It is a contradiction with the assumption that there is no fork. Thus, the width of each layer $|\phi_{i+1}| \leq n$. The theorem is proved.
  
\begin{thm}
	LongestPathLayering(G) and CoffmanGraham(G,n) computes the same layer assignment. That is, $\phi_{LP}$ and $\phi_{CG(n)}$ are the same.
\end{thm} 

\begin{prop}
For each vertex $v \in G$, $\phi_{LP}(v)$ = $\phi_{CG(n)}(v)$.
\end{prop}

In the absence of forks, proofs of the above theorem and proposition are easily obtained, since each layer contains no more $n$ event blocks.

\subsubsection{1/3-BFT}

Now, let us consider the layering result for the network in which has at most $n/3$  dishonest nodes.

Without loss of generality, let $w_p$ be the probability that a node can create fork.
Let $w_c$ be the maximum number of forks a faulty node can create at a time. The number of forked events a node can create at a time is $w_p w_c$.
The number of fork events at each layer is given by 
$$W_{fork} =\frac{1}{3} N w_p w_c$$
Thus, the maximum width of each layer is given by:
$$W_{max} = N + W_{fork} = N + \frac{1}{3} N w_p w_c$$

Thus, if we set the maximum width high enough to tolerate the potential existence of forks, we can achieve BFT of the layering result. The following theorem states the BFT of LPL and CG algorithms.

\begin{thm}
	LongestPathLayering($G$) and CoffmanGraham($G$, $W_{max}$) computes the same layer assignment. That is, $\phi_{LP}$ and $\phi_{CG(W_{max})}$ are identical.
\end{thm} 

\begin{prop}
	For each vertex $v \in G$, $\phi_{LP}(v) = \phi_{CG(W_{max})}(v)$.
\end{prop}

With the assumption about the maximum number of forks at each layer, the above theorem and proposition can be easily proved.

\subsection{Root Graph}\label{se:rootgraph}

We propose to use a new data structure, called \emph{root graph}, which is used to determine frames.

\begin{defn}[Root graph] 
A root graph $G_R$=($V_R$, $E_R$) is a directed graph consisting of vertices as roots and edges represent their reachability.
\end{defn}

In the root graph $G_R$, the set of roots $V_R$ is a subset of $V$. The set of edges $E_R$ is the reduced edges from $E$, in that ($u$,$v$) $\in$ $E_R$ only if $u$ and $v$ are roots and $u$ can reach $v$ following edges in $E$ i.e., $v \hbefore u$.

A root graph contains the $n$ genesis vertices i.e., the $n$ leaf event blocks. 
A vertex $v$ that can reach at least 2/3$n$ + 1 of the current set of all roots $V_R$ is itself a root. For each root $r_i$ that new root $r$ reaches, we include a new edge ($r$, $r_i$) into the set of root edges $E_R$.
Note that, if a root $r$ reaches two other roots $r_1$ and $r_2$ of the same node, then we retain only one edge ($r$, $r_1$) or ($r$, $r_2$) if $\phi(r_1) >\phi(r_2)$. This requirement makes sure each root of a node can have at most one edge to any other node.

Figure~\ref{fig:daglayerrootframe-3} shows an example of the H-OPERA chain, which is the resulting hierarchical graph from a laying algorithm. Vertices of the same layer are placed on a horizontal line. There are 14 layers which are labelled as L0, L1, ..., to L14. Roots are shown in red colors.
 Figure~\ref{fig:rootgraph-3} gives an example of the root graph constructed from the H-OPERA chain.

\begin{figure}
	\centering			
	\subfloat[Hierarchical graph with frame assignment]{\label{fig:daglayerrootframe-3}
		\includegraphics[width=0.28\linewidth]{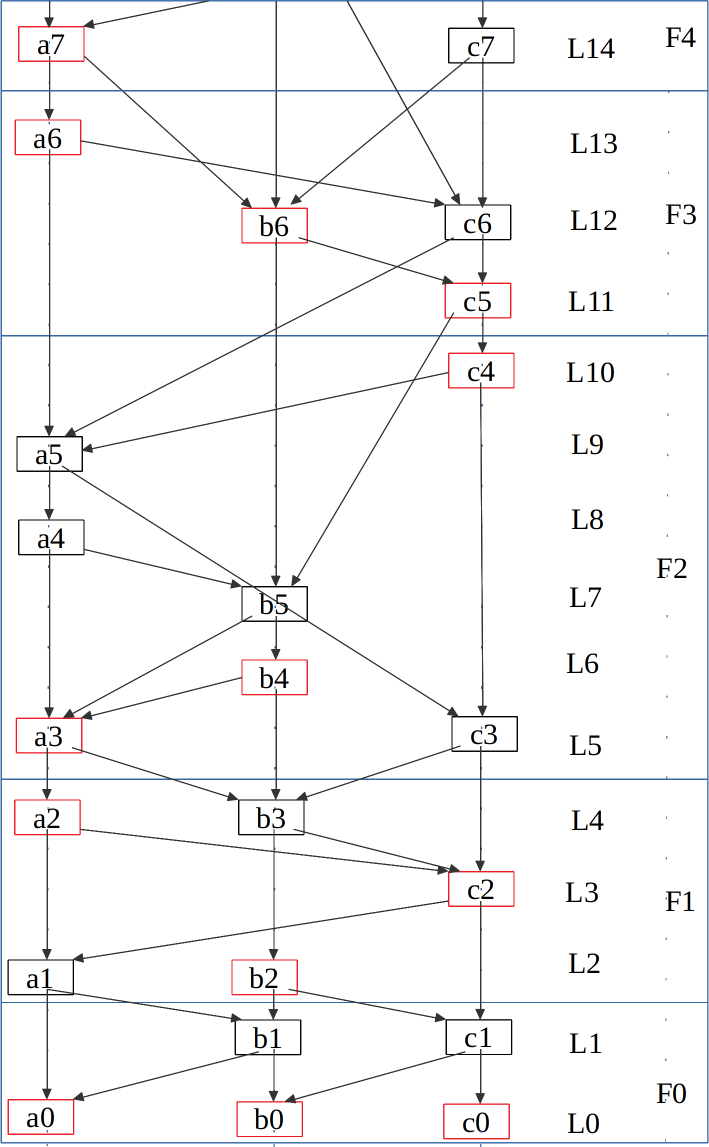}}
	\quad
	\subfloat[Root graph]{\label{fig:rootgraph-3}\includegraphics[width=0.28\linewidth]{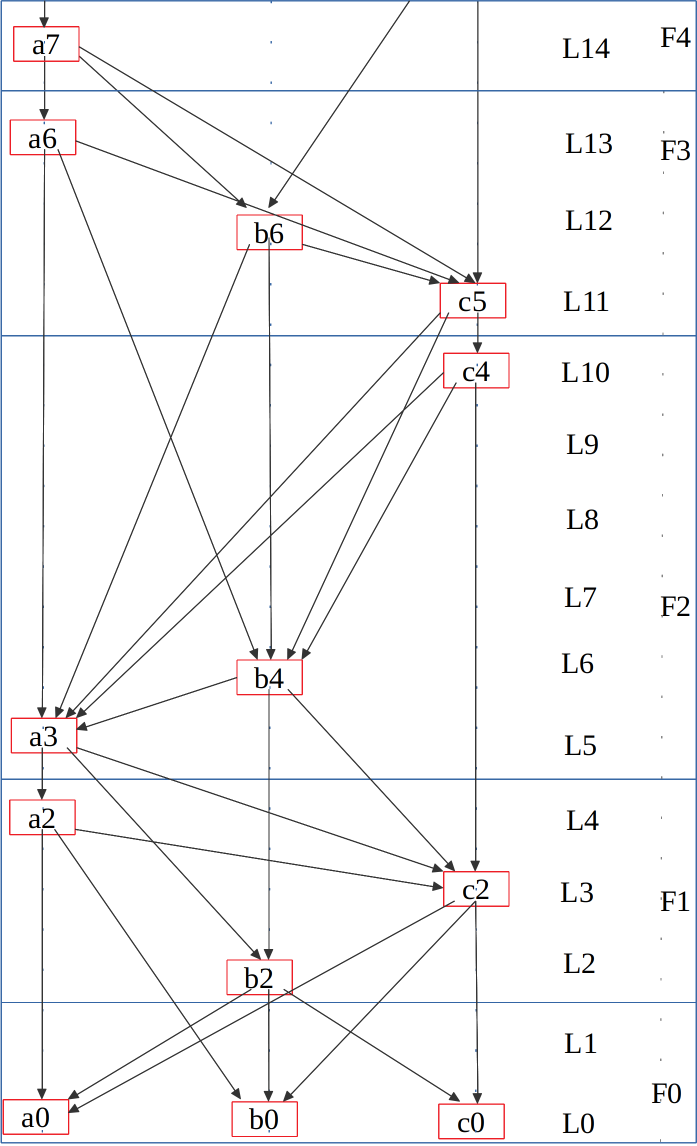}}
	\quad
	\subfloat[Root-layering of a root graph]{\label{fig:rootgraph-3-layer}\includegraphics[width=0.3\linewidth]{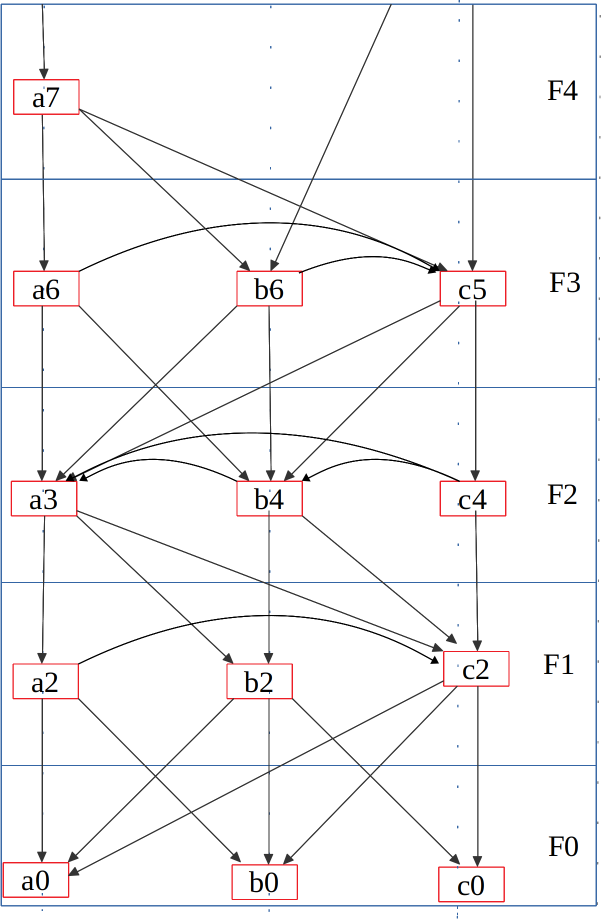}}
	\caption{An example of hierachical graph and root graph computed from the OPERA chain in Figure~\ref{fig:dag-ex}}
	\label{fig:layerframe}	
\end{figure}

Figure~\ref{fig:5} shows a H-OPERA chain (hierarchical graph resulted from layering an OPERA chain) on a single node in a 5-node network.
Figure~\ref{fig:rootgraph-5} depicts the root graph from the H-OPERA chain. There are 14 layers in the example.

\begin{figure}
	\centering
	\subfloat[Hierarchical graph]{\label{fig:5}
	\includegraphics[width=0.27\linewidth]{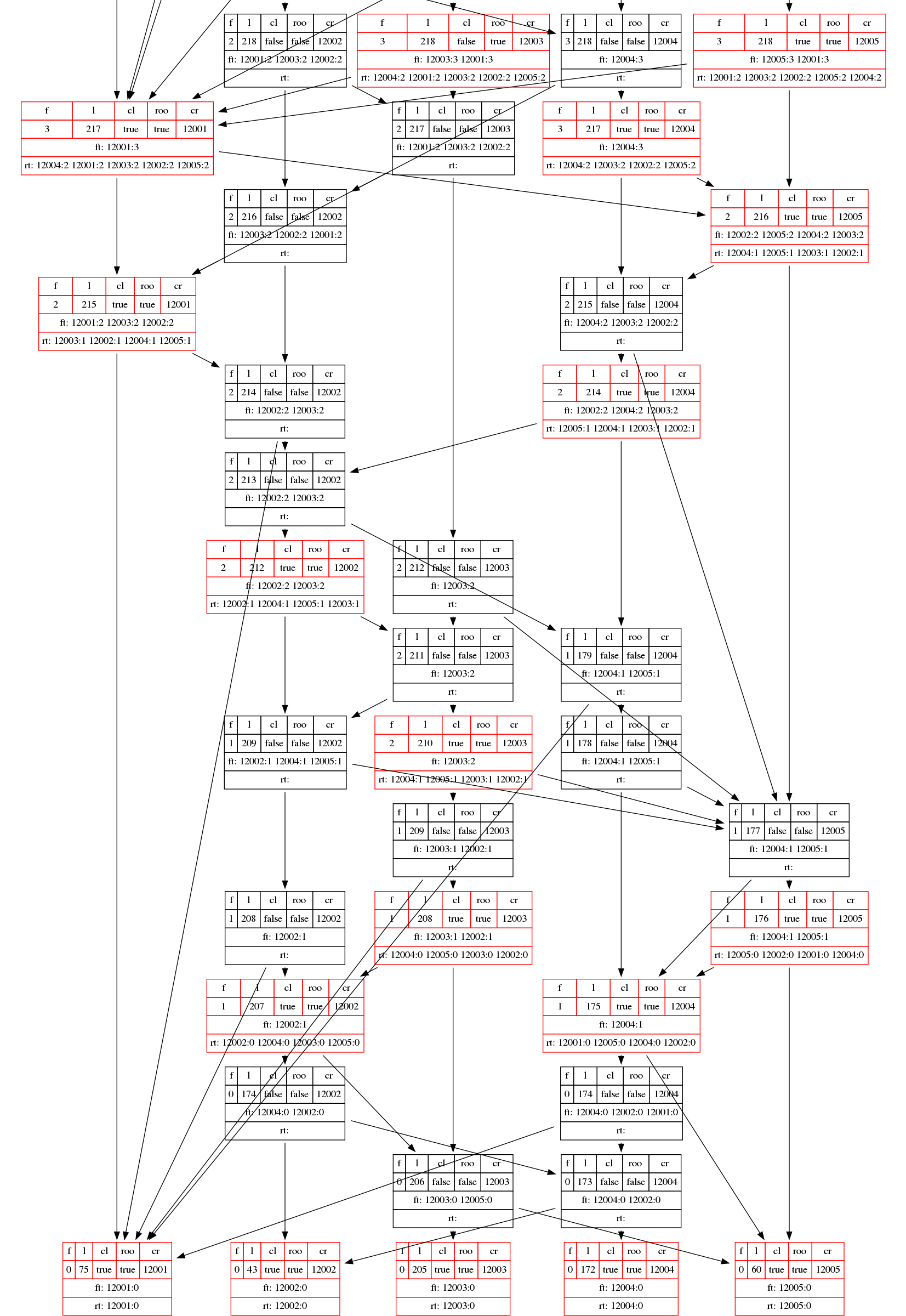}}
\quad
\subfloat[Root graph]{\label{fig:rootgraph-5}
	\includegraphics[width=0.3\linewidth]{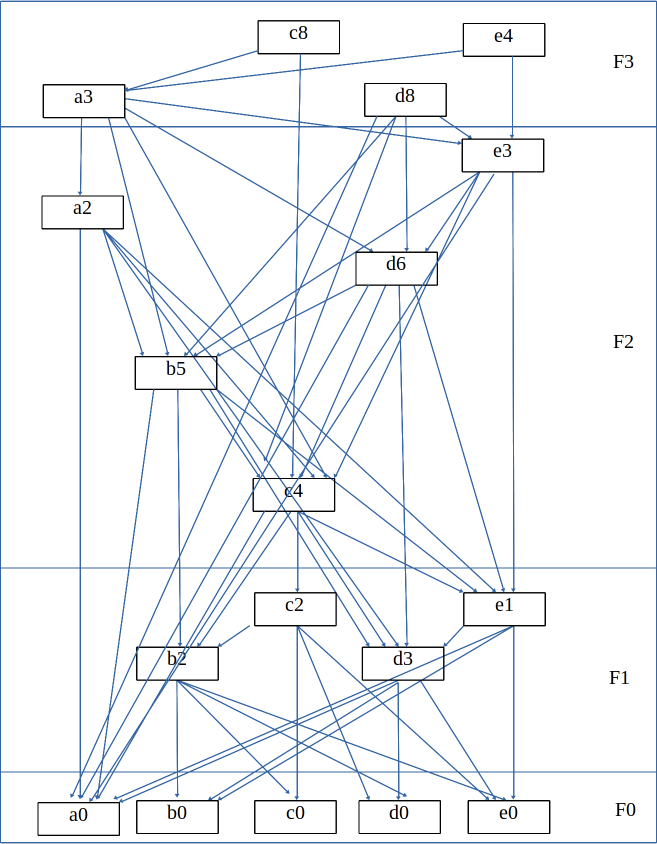}}
\quad
\subfloat[Root-layering of root graph]{\label{fig:rootgraph-5-layer}
	\includegraphics[width=0.35\linewidth]{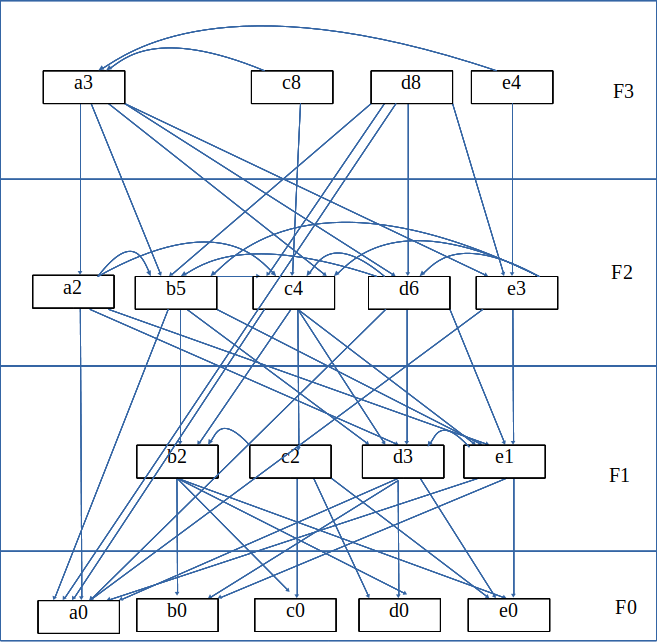}}
	\caption{An example of OPERA chain on a node in a 5-node network. (a) hierarchical graph, (b) root graph, (c) layering root frame for frame assignment.}	
\end{figure}

\subsection{Frame Assignment}\label{se:frameassignment}

We then present a determinstic approach to  frame assignment, which assigns each event block a unique frame number.
First, we show how to assign frames to root vertices via the so-called \emph{root-layering}. Second, we then assign non-root vertices to the determined frame numbers with respect to the topological ordering of the vertices.

\begin{defn}[Root-layering]
A root layering assigns each root with a unique layer (frame) number.	
\end{defn}

Here, we model root layering problem as a modified version of the original layering problem.
For a root graph $G_R=(V_R, E_R)$, root-layering $\phi_R$ is a layering that assigns each root vertex $v_r$ of $V_R$ with an integer $\phi_R(v_r)$ such that:
\begin{itemize}
	\item $\phi_R(u) \geq \phi_R(v) + 1$, for each edge $e=(u,v) \in E_R$.
	\item if $u$ reaches at least $2/3n + 1$ of the roots of frame $i$, then $\phi_R(v)$ = $i +1$.
\end{itemize}

If $\phi_R(v)=j$, then $v$ is a frame-$j$ vertex and $V_j= \phi_R^{-1}(j)$ is the $j$-th frame of $G$. In terms of happened before relation, $v_i \hbefore v_j$ iff $\phi_R(v_i) < \phi_R(v_j)$.
Similarly to the original layering problem, root layering $\phi_R$ partitions the set of vertices $V_R$ into a finite number $l$ of \emph{non-empty} disjoint subsets (called frames) $V_1$,$V_2$,$\dots$, $V_l$, such that $V_R$ = $\cup_{i=1}^{l}{V_i}$. Each vertex is assigned to a root-layer $V_j$, where $1\leq j \leq l$, such that every edge $(u, v) \in E$, $u \in V_i$, $v \in V_j$, $1 \leq i < j \leq l$.

Figure~\ref{fig:rootgraph-3-layer} depicts an example of root layering for the root graph in Figure~\ref{fig:rootgraph-3}. There are five frames assigned to the root graph. The frames are labelled as F0, F1, F2, F3 and F4.
Figure~\ref{fig:rootgraph-5-layer} depicts an example of root layering for the root graph in Figure~\ref{fig:rootgraph-5}.
There are four frames whose labels are F0, F1, F2 and F3.

We now present our new approach that uses root layering information to assign frames to non-root vertices.

\begin{defn}[Frame assignment]
	A frame assignment assigns every vertex $v$ in H-OPERA chain $H$ a frame number $\phi_F(v)$ such that 
	\begin{itemize}
		\item $\phi_F(u) \geq \phi_F(v)$ for every directed edge $(u,v)$ in $H$;
		\item for each root $r$, $\phi_F(r)= \phi_R(r)$;
		\item for every pair of $u$ and $v$: if $\phi(u) \geq \phi(v)$ then $\phi_F(u) \geq \phi_F(v)$;  if $\phi(u) = \phi(v)$ then $\phi_F(u) = \phi_F(v)$.
	\end{itemize}
\end{defn}


Figure~\ref{fig:wholelayerassign} shows an example of the H-OPERA chain with frame assignment of the OPERA chain in Figure~\ref{fig:daglayerrootframe-3}. Each vertex shows the assigned frame number, Lamport timestamp, consensus timestamp, flags to show whether they are atropos, root or clotho, and its creator's identifier.
\begin{figure}
	\centering
	\includegraphics[width=0.3\linewidth]{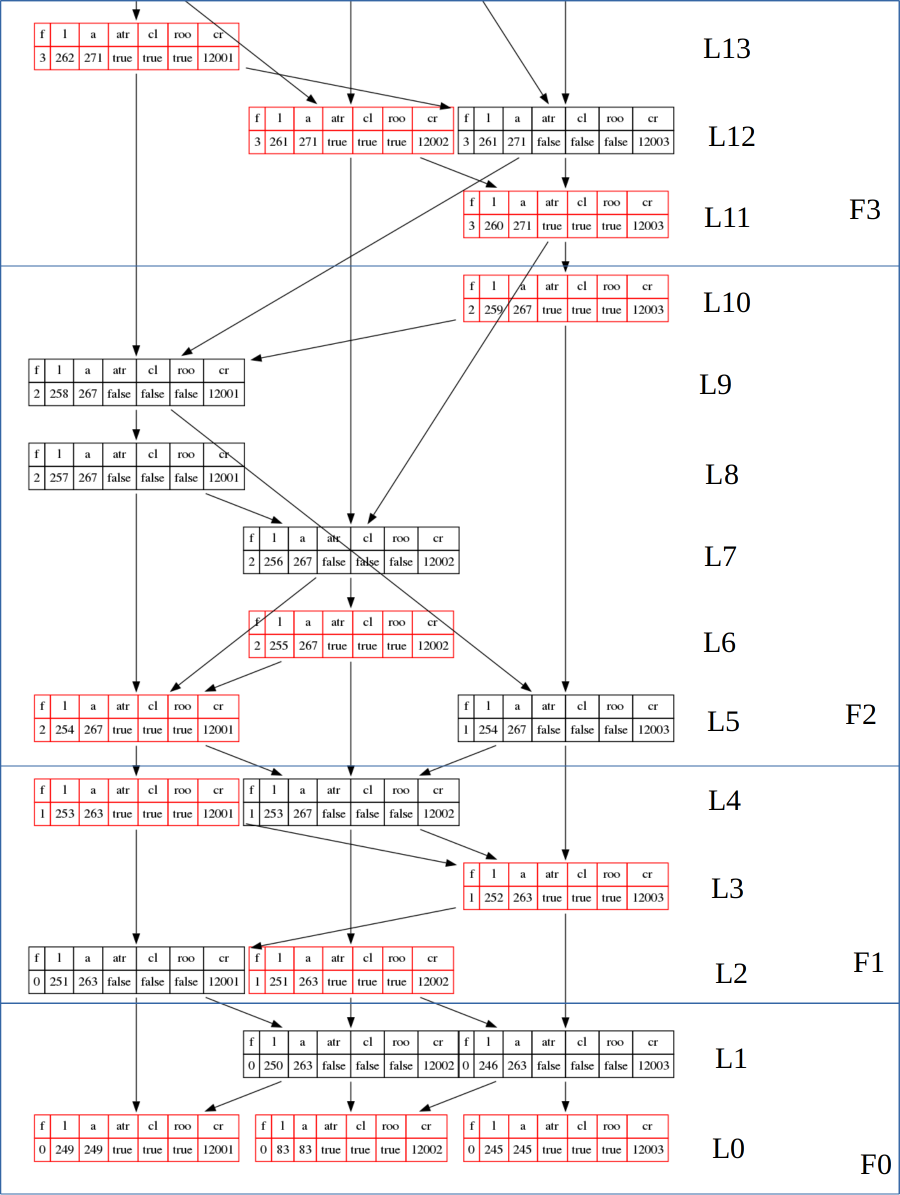}
	\caption{An example of H-OPERA with frame assignment in $L_\phi$ protocol.}
	\label{fig:wholelayerassign}
\end{figure}

\subsection{Consensus}

The protocol $L_{\phi}$ uses several novel concepts such as OPERA chain, H-OPERA chain, Root graph, Frame assignment to define a deterministic solution for the consensus problem.
We now present our formal model for the consistency of knowledge across the distributed network of nodes. We use the consistent cut model, described in Section~\ref{sec:cck}. For more details, see the CCK paper~\cite{cck92}.

For an OPERA chain $G$, let $G[v]$ denote the subgraph of $G$ that contains nodes and edges reachable from $v$.

\begin{defn}[Consistent chains]
	For two chains $G_1$ and $G_2$, they are consistent if for any event $v$ contained in both chains, $G_1[v] = G_2[v]$. Denoted by $G_1 \sim G_2$.
\end{defn}

For any two nodes, if their OPERA chains contain the same event $v$, then they have the same $k$ hashes contained within $v$. 
A node must already have the $k$ references of $v$ in order to accept $v$. Thus, both OPERA chains must contain $k$ references of $v$. Presumably, the cryptographic hashes are  always secure and thus references must be the same between nodes. By induction, all ancestors of $v$ must be the same. When two consistent chains contain the same event $v$, both chains contain the same set of ancestors for $v$, with the same reference and self-ref edges between those ancestors. Consequently, the two OPERA chains are consistent.\\

\begin{defn}[Global OPERA chain]
	A global consistent chain $G^C$ is a chain such that $G^C \sim G_i$ for all $G_i$.
\end{defn}

Let $G \sqsubseteq G'$ denote that $G$ is a subgraph of $G'$. Some properties of $G^C$ are given as follows:
\begin{enumerate}[(1)]
	\item $\forall G_i$ ($G^C \sqsubseteq G_i$).
	\item
	$\forall v \in G^C$ $\forall G_i$ ($G^C[v] = G_i[v]$).
	\item
	($\forall v_c \in G^C$) ($\forall v_p \in G_i$) (($v_p \hbefore v_c) \Rightarrow v_p \in G^C$).
\end{enumerate}

The layering of consistent OPERA chains is consistent itself.
\begin{defn}[Consistent layering]
	For any two consistent OPERA chains $G_1$ and $G_2$,  layering results $\phi^{G_1}$ and $\phi^{G_2}$ are consistent if $\phi^{G_1}(v) = \phi^{G_2}(v)$, for any  vertex $v$ common to both chains.	 Denoted by $\phi^{G_1} \sim \phi^{G_2}$.
\end{defn}

\begin{thm}	
	For two consistent OPERA chains $G_1$ and $G_2$, the resulting H-OPERA chains using layering $\phi_{LPL}$ are consistent.
\end{thm} 

The theorem states that for any event $v$ contained in both OPERA chains, $\phi_{LPL}^{G_1}(v) = \phi_{LPL}^{G_2}(v)$. Since $G_1 \sim G_2$, we have $G_1[v]= G_2[v]$. Thus, the height of $v$ is the same in both $G_1$ and $G_2$. Thus, the assigned layer using $\phi_{LPL}$ is the same for $v$ in both chains.

\begin{prop}[Consistent root graphs]
	Two root graphs $G_R$ and $G'_R$ from two consistent H-OPERA chains are consistent.
\end{prop}

\begin{defn}[Consistent root]
	Two chains $G_1$ and $G_2$ are root consistent, if for every $v$ contained in both chains, $v$ is a root of $j$-th frame in $G_1$, then $v$ is a root of $j$-th frame in $G_2$.
\end{defn}

\begin{prop}
	For any two consistent OPERA chains $G_1$ and $G_2$, they are root consistent. 
\end{prop}

By consistent chains, if $G_1 \sim G_2$ and $v$ belongs to both chains, then $G_1[v]$ = $G_2[v]$.
We can prove the proposition by induction. For $j$ = 0, the first root set is the same in both $G_1$ and $G_2$. Hence, it holds for $j$ = 0. Suppose that the proposition holds for every $j$ from 0 to $k$. We prove that it also holds for $j$= $k$ + 1.
Suppose that $v$ is a root of frame $f_{k+1}$ in $G_1$. 
Then there exists a set $S$ reaching 2/3 of members in $G_1$ of frame $f_k$ such that $\forall u \in S$ ($u\hbefore v$). As $G_1 \sim G_2$, and $v$ in $G_2$, then $\forall u \in S$ ($u \in G_2$). Since the proposition holds for $j$=$k$, 
As $u$ is a root of frame $f_{k}$ in $G_1$, $u$ is a root of frame $f_k$ in $G_2$. Hence, the set $S$ of 2/3 members $u$ happens before $v$ in $G_2$. So $v$ belongs to $f_{k+1}$ in $G_2$.

Thus, all nodes have the same consistent root sets, which are the root sets in $G^C$. Frame numbers are consistent for all nodes.\\




\begin{prop}[Consistent Clotho]
A root $r_k$ in the frame $f_{a+3}$ can nominate a root $r_a$ as Clotho if more than 2n/3 roots in the frame $f_{a+1}$ dominate $r_a$ and $r_k$ dominates the roots in the frame $f_{a+1}$.
\end{prop}


\begin{prop}[Consistent Atropos]
An Atropos is a Clotho that is decided as final. Event blocks in the subgraph rooted at the Atropos are also final events. Atropos blocks form a Main-chain, which allows time consensus ordering and responses to attacks.
\end{prop}

\begin{prop}[Consistent Main-chain]
 \emph{Main-chain} is a special sub-graph of the OPERA chain that stores  Atropos vertices.	
\end{prop}

An event block is called a \emph{root} if the event block is linked to more than two-thirds of previous roots. A leaf vertex is also a root itself. With root event blocks, we can keep track of ``vital'' blocks that $2n/3$ of the network agree on.  
Each participant node computes to the same Main chain, which is the consensus of all nodes, from its own event blocks.


\subsection{Topological sort}\label{se:toposort}

In this section, we present an approach to ordering event blocks that reach finality. The new approach is deterministic subject to any deterministic layering algorithm. 

After layers are assigned to achieve H-OPERA chian, root graphs are constructed. Then frames are determined for every event block. Once certain roots become Clothos, which are known by by 2/3$n$ + 1 of the roots that are in turn known by 2/3$n$ + 1 of the nodes. The final step is to compute the ordering of event blocks that are final. 

Figure~\ref{fig:crust} shows H-OPERA chain of event blocks. The Clothos/Atropos vertices are shown in red. For each Atropos vertex, we compute the subgraph under it. These subgraphs are shown in different colors. We process Atropos from lower layer to high layer.

\begin{figure}
	\centering
	\includegraphics[width=0.3\linewidth]{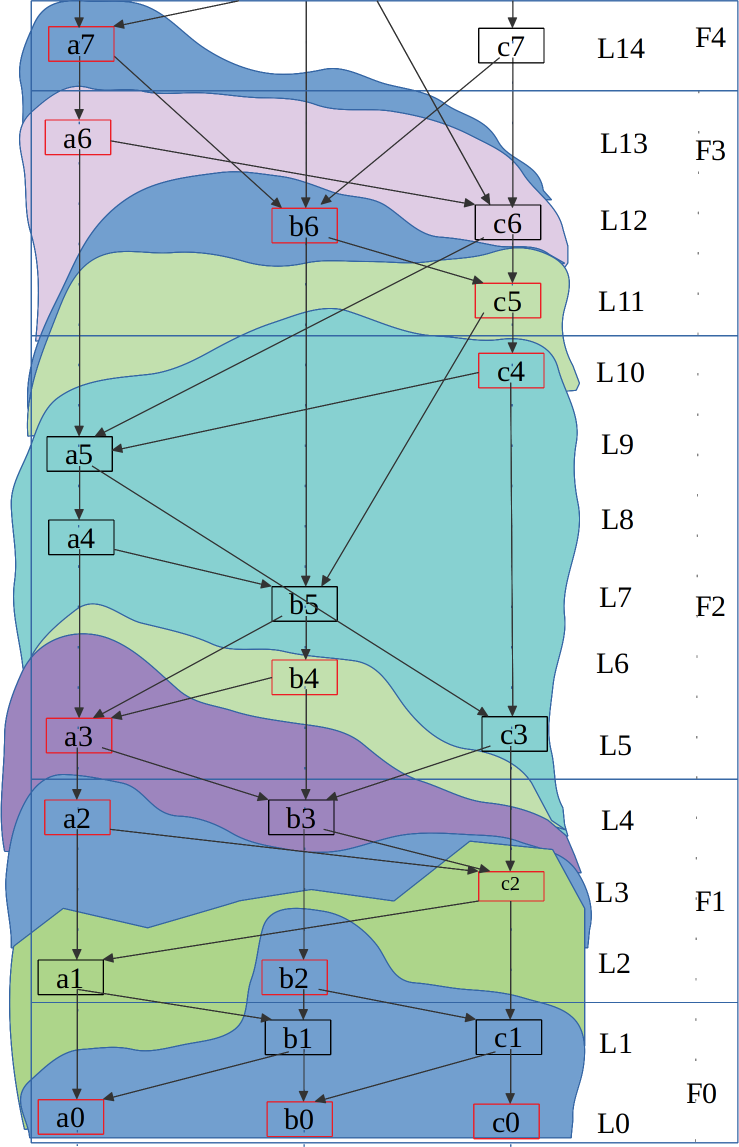}
	\caption{Peeling the subgraphs in topological sorting}
	\label{fig:crust}
\end{figure}

Our algorithm for topological ordering the consensed event blocks is given in Algorithm~\ref{algo:topoordering}.
The algorithm requires that H-OPERA chain and layering assignment  from $\phi$ are precomputed. The algorithm takes as input the set of Atropos vertices, denoted by $A$.
We first order the atropos vertices using $SortVertexByLayer$ function, which sorts vertices based on their layer, then lamport timestamp and then hash information of the event blocks. Second, we process every Atropos in the sorted order, and then compute the subgraph $G[a] = (V_a, E_a)$ under each Atropos. The set of vertices $V_u$ contains vertices from $V_a$ that is not yet process in $U$.
We then apply $SortVertexByLayer$ to order vertices in $V_u$ and then append the ordered result into the final ordered list $S$.

\begin{algorithm}[H]
	\caption{TopoSort}\label{algo:topoordering}
	\begin{algorithmic}[1]
		\State Require: H-OPERA chain $H$, $\phi$, $\phi_F$.
		\Function{TopoSort}{$A$}
		\State $S \leftarrow$ empty list
		\State $U \leftarrow \emptyset$
		\State $Q \leftarrow$ SortVertexByLayer($A$)
		\For{Atropos $a \in Q$}
		\State Compute the graph $G[a] = (V_a, E_a)$
		\State $V_u \leftarrow V_a \setminus U$
		\State $T \leftarrow SortVertexByLayer(V_u)$
		\State Append $T$ into the end of the ordered list $S$.
		\EndFor
		\EndFunction
	
		\Function{SortVertexByLayer}{V}		
		\State Sort the vertices in $V_u$ by layer, lamport timestamp and hash in that order.
		\EndFunction
	\end{algorithmic}
\end{algorithm}

With the final ordering computed using the above algorithm, we can assign the consensus time to the finally ordered event blocks.

\newpage
\section{\onlay\ Framework}\label{se:onlay}

We now present our \onlay\ framework, which is a practical DAG-based solution to distributed ledgers.

Algorithm~\ref{al:main} shows the main function, which is the entry point to launch \onlay\ framework. There are two main loops in the main function and they run asynchronously in parallel. The first loop attempts to create new event blocks and then communicate with other nodes about the new blocks.
The second loop will accept any incoming sync request from other nodes. The node will retrieve the incoming event blocks and will send responses that consist of its known events.

\begin{algorithm}[H]
\caption{\onlay\ Main Function}\label{al:main}
\begin{algorithmic}[1]
	\Function{Main Function}{}
	\BState \emph{loop}:
	\State $\{n_i\}$ $\leftarrow$ $k$-PeerSelectionAlgo()
	\State Request sync to each node in $\{n_i\}$
	\State (SyncPeer) all known events to each node in $\{n_i\}$
	\State Create new event block: newBlock($n$, ${n_i}$)
	\State (SyncOther) Broadcast out the message
	\State Apply layering to obtain H-OPERA chain
	\State Root selection
	\State Compute root graph
	\State Compute global state
	\State Clotho selection
	\State Atropos selection
	\State Order vertices (and assign consensus time)
	\State Compute Main chain
	\BState \emph{loop}:
	\State Accepts sync request from a node
	\State Sync all known events by Lachesis protocol
	\EndFunction
\end{algorithmic}
\end{algorithm}

Specifically, in the first loop, a node makes synchronization requests to $k$ other nodes and it then creates event blocks. Once created, the blocks are broadcasted to all other nodes. 
In line 3, a node runs the Node Selection Algorithm, to select $k$ other nodes that it will need to communicate with. In line 4 and 5, the node sends synchronization requests to get the latest OPERA chain from the other nodes. 
Once receiving the latest event blocks from the responses, the node creates new event blocks (line 6). The node then broadcasts the created event block to all other nodes (line 7). 
After new event block is created, the node updates its OPERA chain, then will then apply layering (line 8). After layering is performed, H-OPERA chain is obtained. It then checks whether the block is a root (line 9) and then computes the root graph (line 10). In line 11, it computes global state from its H-OPERA chain. Then the node decides which roots are Clothos (line 12) and which then becomes Atropos (line 13). Once Atropos vertices are confirmed, the algorithm runs a topological sorting algorithm to get the final ordering for vertices at consensus (line 14).
The main chain is constructed from the Atropos vertices that are found (line 15).

\subsection{Peer selection algorithm}

In order to create a new event block, a node needs to synchronize with $k$ other nodes to get their latest top event blocks. The peer selection algorithm computes such set of $k$ nodes that a node will need to make a synchronization request.

There are multiple ways to select $k$ nodes from the set of $n$ nodes. An simple approach can use a random selection from the pool of $n$ nodes. A more complex approach is to define a cost model for the node selection. In general, $L_{\phi}$ protocol does not depend on how peer nodes are selected.

In \onlay, we have tried a few peer selection algorithms, which are described as follows:
(1) Random: randomly select a peer from $n$ peers;
(2) Least Used: select the least use peer(s).
(3) Most Used (MFU): select the most use peer(s).
(4) Fair: select a peer that aims for a balanced distribution.
(5) Smart: select a peer based on some other criteria, such as successful throughput rates, number of own events.

\subsection{Peer synchronization}
Now we describe the steps to synchronize events between the nodes, as presented in Algorithm~\ref{al:syncevents}.
Each node $n_1$ selects a random peer $n_2$ and sends a sync request indicating the local known events of $n_1$ to $n_2$. Upon receiving the sync request from $n_1$, $n_2$ will perform an event diff with its own known events and will then return the unknown events to $n_1$.

\begin{algorithm}[H]
	\caption{EventSync}\label{al:syncevents}
	\begin{algorithmic}[1]
		\Function{sync-events()}{}		
		\State $n_1$ selects a random peer $n_2$ to synchronize with
		\State $n_1$ gets local known events (map[int]int)
		\State $n_1$ sends a RPC Sync request to peer
		\State $n_2$ receives a RPC Sync request
		\State $n_2$ does an EventDiff check on the known map (map[int]int)		
		\State $n_2$ returns the unknown events to $n_1$		
		\EndFunction
	\end{algorithmic}
\end{algorithm}

The algorithm assumes that a node always needs the events in topological ordering (specifically in reference to the lamport timestamps). Alternatively, one can simply use a fixed incrementing index or layer assignment to keep track of the top event for each node.


\subsection{Node Structure}
This section gives an overview of the node structure in \onlay. A node in \onlay\ system is a machine that participates with other nodes in the network.

\begin{figure}
	\centering
	\includegraphics[width=0.5\linewidth]{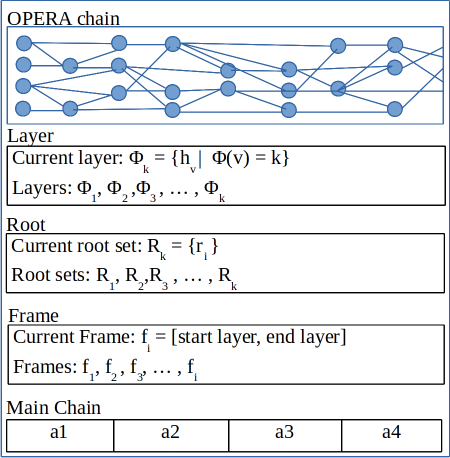}
	\caption{\onlay\ node structure}
	\label{fig:nodestructure}
\end{figure}

Figure~\ref{fig:nodestructure} shows the structure of a node. 
Each node consists of an OPERA chain of event blocks, which is the local view of the DAG. In each node, it stores the layering information (H-OPERA chain), roots, root graph, frames, Clotho set and Atropos set.
The top event block indicates the most recently created event block by this node.
Each frames contains the root set of that frame. A root can become a Clotho if it is known by 2/3$n$ + 1 of the roots, which are then known by another root. When a root becomes Clotho, its Clotho flag is set on. A Clotho is then promoted to an Atropos, once it is assigned with a final ordering number. The Main-chain is a data structure storing hash values of the Atropos blocks. 

\subsection{Event block creation}
Like previous Lachesis protocol, $L_\phi$ protocol allows every node to create an event block. A new event block refers to the top event blocks of $k$ other nodes by using their hash values. 
One of the references is a self-ref that references to an event block of the same node. 
The other $k$-1 references refer to the top event nodes of other $k$ -1 nodes.
An event block can only be created if the referencing event blocks exist.

\subsection{Root selection}

There is a Root selection algorithm given in Fantom paper~\cite{fantom18}.
In \onlay, we propose a new approach that uses root graph and frame assignment (see Section~\ref{se:rootgraph} and Section~\ref{se:frameassignment}).
 
The root graph $G_R = (V_R, E_R)$ contains vertices as  roots $V_R  \subseteq V$, and the set of edges $E_R$  the reduced edges from $E$, such that $(u,v) \in E_R$ only if $u$ and $v$ are roots and there is a path from $u$ to $v$ following edges in $E$.

A root graph contains the $n$ genesis vertices i.e., the $n$ leaf event blocks. 
A vertex $v$ that can reach at least 2/3$n$ + 1 of the current set of all roots $V_R$ is itself a root. For each root $r_i$ that new root $r$ reaches, we include a new edge ($r$, $r_i$) into the set of root edges $E_R$.
Note that, if a root $r$ reaches two other roots $r_1$ and $r_2$ of the same node, and $\phi(r_1) >\phi(r_2)$ then we include only one edge ($r$, $r_1$) in $E_R$. This contraint ensures that each root of a node can have at most one edge to any other node.

The steps to build the root graph out of an H-OPERA chain is given in Algorithm~\ref{al:buildrootgraph}.
The algorithm processes from layer 0 to maximum layer $l$. Line 3, $R$ is the set of active roots, we will update the set after finding new roots at each layer. Line 4 and 5 set the initial value of the vertex set and edge set of the root graph. For each layer $i$, it runs the steps from line 8 to 18. $Z$ is the set of new roots at layer $i$, in which $Z$ is initially set to empty (Line 8). We then process each vertex $v$ in the layer $\phi_i$ (line 9). We compute the set $S$ of all active roots in $R$ that $v$ can reach. If more than 2/3$n$ of $R$ member is reached, we can promote $v$ to be a new root (line 11). For every root $v_i$ in $S$ that is reachable from $v$, we add an add $(v,v_i)$ in the $E_R$.
In line 14-15, when $v$ becomes a new root, we add $v$ into $V_R$ of the root graph, and also into $Z$ which keeps track of the new roots found at layer $i$. 
The second inner loop will update the current active set $R$ with the new roots found at layer $i$ (line 16-18).

\begin{algorithm}[H]
	\caption{Root graph algorithm}\label{al:buildrootgraph}
	\begin{algorithmic}[1]
		\State Require: H-OPERA chain
		\State Output: root graph $G_R=(V_R,E_R)$
		\State{$R \leftarrow$ set of leaf events}
		\State $V_R \leftarrow R$
		\State $E_R \leftarrow \emptyset$

		\Function{buildRootGraph}{$\phi$, $l$}
		\For{each layer $i$=1..$l$}
			\State$Z \leftarrow \emptyset$
			\For{each vertex $v$ in layer $\phi_i$}
				\State $S \leftarrow$ the set of vertices in $R$ that $v$ reaches
				\If{$|S| > 2/3n$}
				\For{each vertex $v_i$ in $S$}
					\State $E_R \leftarrow E_R \cup \{(v,v_i)\}$
				\EndFor
				\State $V_R \leftarrow V_R \cup \{v\}$
				\State $Z \leftarrow Z \cup \{v\}$ 
				\EndIf
			\EndFor
			
			\For{each vertex $v_j$ in $Z$}
				\State Let $v_{old}$ be a root in $R$ such that $cr(v_j) = cr(v_{old})$		
				\State $R \leftarrow R \setminus \{v_{old}\} \cup \{v_j\}$
			\EndFor
		\EndFor
		\EndFunction
	\end{algorithmic}
\end{algorithm}

The algorithm always updates the active set $R$ of roots at each layer. The number of active roots in $R$ is always $n$, though the active roots can be from different layers.

After the root graph is constructed, we can assign frame numbers to the roots as specified in  Section~\ref{se:frameassignment}). Frames are assigned to the roots using root-layering algorithm.
When a new root reaches more than 2/3$n$ of the roots of a frame $i$, the root is assigned to the next frame.


\subsection{Fork detection and removal}

When a node $n$ receives Sync request from another node, $n$ will check if the received events would cause a fork. In a 1/3-BFT, we proved in our previous paper~\cite{fantom18} that there exists an honest node that sees the fork. It will remove one of the forked events from its OPERA chain, and then notifies the fork to all other nodes.

When a root is selected, it is certain that if there exists a fork in the subgraph under the root event block, the fork should be detected and removed.

\begin{prop}[Fork-free]
For any Clotho $v_c$, the subgraph $G[v_c]$ is fork-free.
\end{prop}

When a root $r$ becomes a Clotho, the subgraph $G[r]$ has already detected forks, if any. The fork was already removed and the node already synchronized with other nodes about the fork. By definition of Clotho, the $r$ is known by more than 2/3$n$ events, which are in turn known by another root event.
Thus, more than 2/3$n$ knew about and removed the detected forks that were synchronized.

\begin{prop}[Fork-free Global chain]
	The global consistent OPERA chain $G^C$ is fork-free.
\end{prop}

Recall that the global consistent chain $G^C$ consists of the finalised events i.e., the Clothos and its subgraphs that get final ordering at consensus. Those finally ordered Clothos are Atropos vertices.
From the above proposition, for every Clotho $v_c$, there is no fork in its subgraph  $G[v_c]$. Thus, we can prove by processing the Clothos in topological order from lowest to highest. 

\subsection{Compute global state}

From the local view H-OPERA chain of its own, a node can estimate the subgraph of $H$ that other nodes know at the time. The intuition is that a node $n_i$, from its history $H$, conservatively knows that the other node $n_j$ may just know the event blocks under the top event of $n_j$.
Let $t_i$ denote the top event of a node $n_i$ in the local chain $H$. From the local state $H$, a node can determine all top events of all nodes in $H$.

Algorithm~\ref{al:globalstate} shows an algorithm that a node can estimate a global state from a local state of a node.  The steps show how to compute the global state $H^C$ from the local state $H$. 
The set of processed nodes $S$ is initially empty (line 2). Let $L$ be the stack of all the layers so far.
The processed set $S$ has less than $n$ nodes, we process the layers one at a time in the while loop.  We pop the stack $L$ to get the highest layer $L_i$ (line 4). For each block $v$ in $L_i$, $v$ is removed from $L_i$ if the the creator of $v$ is already processed (line 5-7). After that, the set $L_i$ contains only the event blocks whose creators do not belong to $S$. Line 8 randomly picks a block $v$ from the remainder in $L_i$. Then the processed set $S$ is appended with the creator of $v$. Once all nodes are processed, the global state of the node is the induced subgraph from the remainder nodes in $L$, which excludes the nodes removed by the algorithm.

\begin{algorithm}[H]
	\caption{Estimate a Global State from Local State}\label{al:globalstate}
	\begin{algorithmic}[1]
		\Procedure{globalstate(G,L)}{}
		\State S $\leftarrow \emptyset$				
		\While{$|S| < n$}
			\State $L_i$ $\leftarrow$  pop($L$)
			\For{$v \in L_i$}
				\If{$cr(v) \in S$}
					\State{Remove $v$ from $L_i$}
				\EndIf
			\EndFor
			\State {Remove $v$ (randomly) from $L_i$}
			\State{ $S \leftarrow S \cup {cr(v)}$}
		\EndWhile
		\State{$G^C$ $\leftarrow$ induced graph from remainder nodes in $L$}
		\EndProcedure
	\end{algorithmic}
\end{algorithm}

When a node $n_i$ has an event block $v$, it knows that node all other nodes $n_j$ may not concurrently know the event block $v$. The above algorithm determines the highest event block $v$ that $n_i$ know for sure all the other nodes, from $n_i$'s view, already knows it. One can repeat the algorithm several times, to determine $n_i$ knows that $n_j$ knows that $n_k$ knows a certain event $v$.

\subsection{Clotho selection}

For practical BFT, consensus is reached once an event block is known by more than 2/3$n$ of the nodes and that information is further known by 2/3$n$ of the nodes.
Similarly, we define the consensus of a root, at which condition the root becomes a Clotho, if it can be reached by more than 2/3$n$ of the roots and the information is known by another 2/3$n$ of the roots.


For a root $r$ at frame $f_i$, if there exists a root $r'$ at a frame $f_j$ such that $j$ $\geq$ $i$+2, then the root $r$ reaches pBFT consensus. The root $r'$ is called the \emph{nominator}, which nominates $r$ to become a Clotho.

\begin{prop}[Global Clotho]
	For any two honest nodes $n_i$ and $n_j$, if $c$ is a Clotho in $H_i$ of $n_i$, then $c$ is also a Clotho in $H_j$ of $n_j$.
\end{prop}

\subsection{Atropos Selection}

After a new Clotho is found, we will assign the final \emph{consensus time} for it. When a Clotho is assigned with a consensus time, it becomes an Atropos.

In order to assign consensus time for an Clotho, Atropos selection process first sorts the Clothos based on topological ordering. Section~\ref{se:toposort} gives more details on the topological ordering and consensus time assignment.

\begin{prop}[Global Atropos]
	For any two honest nodes $n_i$ and $n_j$, if $a$ is an Atropos in $H_i$ of $n_i$, then $a$ is also an Atropos in $H_j$ of $n_j$.
\end{prop}


After Atropos consensus time is computed, the Clotho is nominated to Atropos and each node stores the hash value of Atropos and Atropos consensus time in Main-Chain (blockchain). The Main-chain is used for time order between event blocks. The proof of Atropos consensus time selection is shown in the section~\ref{se:proof}.

\subsection{Transaction confirmations}

Here are some steps for a transaction to reach finality in our system. First, when a user submits a transaction into a node, a successful submission receipt will be issued to the client as a confirmation of the submitted transaction. Second, the node will batch the user submitted transaction(s) into a new event block amended into a DAG structure of event blocks and then will broadcast the event block to all other nodes of the system. The node will send a receipt confirming that the containing event block identifier is being processed. Third, when the event block is known by majority of the nodes (e.g., it becomes a Root event block), or being known by such a Root block, the node will send a confirmation that the event block is acknowledged by a majority of the nodes. Fourth, our system will determine the condition at which a Root event block becomes a Clotho for being further acknowledged by a majority of the nodes. A confirmation is then sent to the client to indicate that the event block has come to the semi-final stage as a Clotho or being confirmed by a Clotho. Fifth, after the Clotho stage, we will determine the consensus timestamp for the Clotho and its dependent event blocks. Once an event block gets the final consensus timestamp, it is finalized and a final confirmation will be issued to the client that the transaction has been successfully finalized. Thus, there will be a total of five confirmations to be sent and the fifth receipt is the final confirmation of a successful transaction.

There are some cases that a submitted transaction can fail to reach finality. Examples include a transaction does not pass the validation, e.g., insufficient account balance, or violation of account rules. The other kind of failure is when the integrity of DAG structure and event blocks is not complied due to the existence of compromised or faulty nodes. In such unsuccessful cases, the event blocks are marked for removal and detected issues are notified to all nodes. Receipts of the failure will be sent to the client.


\section{Conclusion}\label{se:con}
In this paper, we present a new framework, namely \onlay, for a scalable asynchronous distributed system with practical BFT. We propose a new consensus protocol, called $L_\phi$, to address a more reliable consensus compared to predecedent DAG-based approaches.
The new consensus protocol uses the well-known concept of layering of DAG to achieve a deterministic consensus on the OPERA chain.
 
We further propose new approaches to optimizing the OPERA chain such as using H-OPERA, root graph, root layering, improved frame assignment, improved Clotho and Atropos selection, and Main-chain. These new models, data structures and algorithm deliver a more intuitive model and a faster more reliable consensus in a distributed system.

We have presented a formal definitions and semantics for $L_\phi$ protocol in Section~\ref{se:prelim}.
Our formal proof of pBFT for our Lachesis protocol is given in Section~\ref{se:appendix}. Our work extends the formal foundation established in Fantom paper~\cite{fantom18}, which is the first that studies concurrent common knowledge sematics~\cite{cck92} in DAG-based protocols.

\newpage
\section{Appendix}\label{se:appendix}

This section gives further details about the $L_\phi$ protocol and then shows the proof of BFT of the protocol. We then present the formal semantics of  $L_\phi$ using the concurrent common knowledge that can be applied to general model of DAG-based approaches.

\subsection{Basic Definitions}

Each node stores a local view of the history built from the $L_\phi$ protocol. We cover some basic definitions used in our previous Lachesis papers~\cite{lachesis01,fantom18}. We also include definitions of layering, which our \onlay\ framework and $L_\phi$ protocol is based on.

\begin{defn}[node]
	Each machine, which is a participant in the Lachesis protocol, is called a node. \end{defn}

Let $n$ be the number of nodes in a network running \onlay.
\begin{defn}[peer node]
	A node $n_i$ has $n$-1 peer nodes.
\end{defn}

The history can be represented by a directed acyclic graph $G=(V, E)$, where $V$ is a set of vertices and $E$ is a set of edges. Each vertex in a row (node) represents an event. Time flows left-to-right of the graph, so left vertices represent earlier events in history.
A path $p$ in $G$ is a sequence  of vertices ($v_1$, $v_2$, $\dots$, $v_k$) by following the edges in $E$.
Let $v_c$ be a vertex in $G$.
A vertex $v_p$ is the \emph{parent} of $v_c$ if there is an edge from $v_p$ to $v_c$.
A vertex $v_a$ is an \emph{ancestor} of $v_c$ if there is a path from $v_a$ to $v_c$.

\subsubsection{OPERA chain}

\begin{defn}[event block]
	Each node can create event blocks, send (receive) messages to (from) other nodes.
\end{defn}

Suppose a node $n_i$ creates an event $v_c$ after an event $v_s$ in $n_i$.  Each event block has exactly $k$ references. One of the references is self-reference, and the other $k$-1 references point to the top events of $n_i$'s $k$-1 peer nodes.

\begin{defn}[OPERA chain]
	OPERA chain is a DAG graph $G = (V, E)$ consisting of $V$ vertices and $E$ edges. Each vertex $v_i \in V$ is an event block. An edge $(v_i,v_j) \in E$ refers to a hashing reference from $v_i$ to $v_j$; that is, $v_i \erefz v_j$.
\end{defn}

\begin{defn}[vertex]
	An event block is a vertex of the OPERA chain.
\end{defn}

\begin{defn}[Leaf]
	The first created event block of a node is called a leaf event block.
\end{defn}

\begin{defn}[top event]
	An event $v$ is a top event of a node $n_i$ if there is no other event in $n_i$ referencing $v$.
\end{defn}

\begin{defn}[self-ref]
	An event $v_s$ is called ``self-ref" of event $v_c$, if the self-ref hash of $v_c$ points to the event $v_s$. Denoted by $v_c \eself v_s$.
\end{defn}

\begin{defn}[ref]
	An event $v_r$ is called ``ref" of event $v_c$ if the reference hash of $v_c$ points to the event $v_r$. Denoted by $v_c \eref v_r$.
\end{defn}

For simplicity, we can use $\erefz$ to denote a reference relationship (either $\eref$ or $\eself$).

\begin{defn}[self-ancestor]
	An event block $v_a$ is self-ancestor of an event block $v_c$ if there is a sequence of events such that $v_c \eself v_1 \eself \dots \eself v_m \eself v_a $. Denoted by $v_c \eselfancestor v_a$.
\end{defn}

\begin{defn}[ancestor]
	An event block $v_a$ is an ancestor of an event block $v_c$ if there is a sequence of events such that $v_c \erefz v_1 \erefz \dots \erefz v_m \erefz v_a $. Denoted by $v_c \eancestor v_a$.
\end{defn}

For simplicity, we simply use $v_c \eancestor v_s$ to refer both ancestor and self-ancestor relationship, unless we need to distinguish the two cases.

\begin{defn}[creator] If a node $n_x$ creates an event block $v$, then the creator of $v$, denoted by $cr(v)$, is $n_x$.
\end{defn}

We introduce pseudo vertices, \emph{top} and \emph{bot}, of the OPERA chain $G$.
\begin{defn}[pseudo top]
	A pseudo vertex, called top, is the parent of all top event blocks. Denoted by $\top$.
\end{defn}
\begin{defn}[pseudo bottom]
	A pseudo vertex, called bottom, is the child of all leaf event blocks. Denoted by $\bot$.
\end{defn}

With the pseudo vertices, we have $\bot$ happened before all event blocks. Also all event blocks happened before $\top$. That is, for all event $v_i$, $\bot$ $\hbefore$ $v_i$ and $v_i$ $\hbefore$ $\top$.

Here, we introduce a new idea of subgraph rooted at a vertex $v$.

\begin{defn}[subgraph] 
	For a vertex $v$ in a DAG $G$, let $G[v] = (V_v,E_v)$ denote an induced subgraph of $G$ such that $V_v$ consists of all ancestors of $v$ including $v$, and $E_v$ is the induced edges of $V_v$ in $G$.
\end{defn}

\subsubsection{Happened-Before relation}
	\begin{defn}[Happened-Immediate-Before]
		An event block $v_x$ is said Happened-Immediate-Before an event block $v_y$ if $v_x$ is a (self-) ref of $v_y$. Denoted by $v_x \hibefore v_y$.
	\end{defn}
	
	\begin{defn}[Happened-Before]	
		An event block $v_x$ is said Happened-Before an event block $v_y$ if $v_x$ is a (self-) ancestor of $v_y$. Denoted by $v_x \hbefore v_y$.
	\end{defn}
	
	The happens-before relation is the transitive closure of happens-immediately-before.
	An event $v_x$ happened before an event $v_y$ if one of the followings happens: (a) $v_y \eself v_x$, (b) $v_y \eref v_x$,  or (c) $v_y \eancestor v_x$.
	We come up with the following proposition:
	\begin{prop}[Happened-Immediate-Before OPERA]
		$v_x \hibefore v_y$ iff $v_y \erefz v_x$ iff edge $(v_y, v_x)$ $\in E$ of OPERA chain.
	\end{prop}
	\begin{lem}[Happened-Before Lemma]
		$v_x \hbefore v_y$ iff $v_y \eancestor v_x$.
	\end{lem}


	\begin{defn}[concurrent]
		Two event blocks $v_x$ and $v_y$ are said concurrent if neither of them  happened before the other. Denoted by $v_x \concur v_y$.
	\end{defn}

	Given two vertices $v_x$ and $v_y$ both contained in two OPERA chains $G_1$ and $G_2$ on two nodes. We have the following: 
	(1) $v_x \hbefore v_y$ in $G_1$ iff $v_x \hbefore v_y$ in $G_2$; (2)
	$v_x \concur v_y$ in $G_1$ iff $v_x \concur v_y$ in $G_2$.

\subsubsection{Layering Definitions}

For a directed acyclic graph $G$=($V$,$E$), a layering is to assign a layer number to each vertex in $G$.

\dfnn{Layering}{A layering (or levelling) of $G$ is a topological numbering $\phi$ of $G$, $\phi: V \rightarrow Z$,  mapping the  set  of  vertices $V$ of $G$ to  integers  such  that $\phi(v)$ $\geq$ $\phi(u)$ + 1 for every directed edge ($u$, $v$) $\in E$.  If $\phi(v)$=$j$, then $v$ is a layer-$j$ vertex and $V_j= \phi^{-1}(j)$ is the jth layer of $G$.}

A layering $\phi$ of $G$ partitions of the set of vertices $V$ into a finite number $l$ of \emph{non-empty} disjoint subsets (called layers) $V_1$,$V_2$,$\dots$, $V_l$, such that $V$ = $\cup_{i=1}^{l}{V_i}$. Each vertex is assigned to a layer $V_j$, where $1 \leq j \leq l$, such that every edge ($u$,$v$) $\in E$, $u \in V_i$, $v \in V_j$, $1 \leq i < j \leq l$.

\begin{defn}[Hierarchical graph]
	For a layering $\phi$, 
	the produced graph $H$=($V$,$E$,$\phi$) is a \emph{hierarchical graph}, which is also called an $l$-layered directed graph and could be represented as
	$H$=($V_1$,$V_2$,$\dots$,$V_l$;$E$).
\end{defn}

\subsection{Proof of Byzantine Fault Tolerance for $L_\phi$ Algorithm}\label{se:proof}
This section presents a proof of our $L_\phi$. We aim to show that our consensus is Byzantine fault tolerant when at most one-third of participants are compromised. We first provide some definitions, lemmas and theorems. Then we validate the Byzantine fault tolerance. 


We introduce a new notion of H-OPERA chain, which is built on top of the OPERA chain.
By applying a layering $\phi$ on the OPERA chain, one can obtain the hierarchical graph of $G$, which is called H-OPERA chain.

\begin{defn}[H-OPERA chain]
	An H-OPERA chain is the result hierarchical graph $H = (V,E, \phi)$.
\end{defn}

\begin{defn}[Root]
	\label{def:root}
	The leaf event block of a node is a root.
	When an event block $v$ can reach more than $2n/3$ of the roots in the previous frames, $v$ becomes a root.
\end{defn}

\begin{defn}[Root set]
	The set of all first event blocks (leaf events) of all nodes form the first root set $R_1$ ($|R_1|$ = $n$). The root set $R_k$ consists of all roots $r_i$ such that $r_i$ $\not \in $ $R_i$, $\forall$ $i$ = 1..($k$-1) and $r_i$ can reach more than 2n/3 other roots in the current frame, $i$ = 1..($k$-1).  
\end{defn}

\begin{defn}[Root graph] 
	A root graph $G_R$=($V_R$, $E_R$) is a directed graph consisting of vertices as roots and edges represent their reachability.
\end{defn}

In the root graph $G_R$, the set of roots $V_R$ is a subset of $V$. The set of edges $E_R$ is the reduced edges from $E$, in that ($u$,$v$) $\in$ $E_R$ only if $u$ and $v$ are roots and $u$ can reach $v$ following edges in $E$ i.e., $v \hbefore u$.

\begin{defn}[Frame]
Frame $f_i$ is a natural number that separates Root sets. 
\end{defn} 

The root set at frame $f_i$ is denoted by $R_i$.

\begin{defn}[consistent chains]\label{dfn:conchains} OPERA chains $G_1$ and $G_2$ are consistent iff for any event $v$ contained in both chains, $G_1[v] = G_2[v]$. Denoted by $G_1 \sim G_2$.
\end{defn}
When two consistent chains contain the same event $v$, both chains contain the same set of ancestors for $v$, with the same reference and self-ref edges between those ancestors:
\begin{thm}\label{thm:conchains}
	All nodes have consistent OPERA chains.
\end{thm}
\begin{proof}
	 If two nodes have OPERA chains containing event $v$, then they have the same $k$ hashes contained within $v$. A node will not accept an event during a sync unless that node already has $k$ references for that event, so both OPERA chains must contain $k$ references for $v$. The cryptographic hashes are assumed to be secure, therefore the references must be the same. By induction, all ancestors of $v$ must be the same. Therefore, the two OPERA chains are consistent.
\end{proof}

\begin{defn}[fork]
	The pair of events ($v_x$, $v_y$) is a fork if $v_x$ and $v_y$ have the same creator, but neither is a self-ancestor of the other. Denoted by $v_x \efork v_y$.
\end{defn}
For example, let $v_z$ be an event in node $n_1$ and two child events $v_x$ and $v_y$ of $v_z$. if $v_x \eself v_z$, $v_y \eself v_z$, $v_x \not \eself v_y$, $v_y \not \eself v_z$, then ($v_x$, $v_y$) is a fork.
The fork relation is symmetric; that is $v_x \efork v_y$ iff $v_y \efork v_x$.
\begin{lem}
	$v_x \efork v_y$ iff $cr(v_x)=cr(v_y)$ and $v_x \concur v_y$.
\end{lem}
\begin{proof}
By definition, ($v_x$, $v_y$) is a fork if $cr(v_x)=cr(v_y)$, $v_x \not \eancestor v_y$ and $v_y \not \eancestor v_x$. Using Happened-Before, the second part means $v_x \not \rightarrow v_y$ and $v_y \not \rightarrow v_x$. By definition of concurrent, we get $v_x \concur v_y$.
\end{proof}

\begin{lem} (fork detection). If there is a fork $v_x \efork  v_y$, then $v_x$ and $v_y$ cannot both be roots on honest nodes.
\end{lem}
\begin{proof}
	Here, we show a proof by contradiction. Any honest node cannot accept a fork so $v_x$ and $v_y$ cannot be roots on the same honest node. Now we prove a more general case. Suppose that both $v_x$ is a root of $n_x$ and $v_y$ is root of $n_y$, where $n_x$ and $n_y$ are honest nodes. Since $v_x$ is a root, it reached events created by more than 2/3 of member nodes. Similary, $v_y$ is a root, it reached events created by  more than 2/3 of member nodes. Thus, there must be an overlap of more than $n$/3 members of those events in both sets. Since we assume less than $n$/3 members are not honest, so there must be at least one honest member in the overlap set. Let $n_m$ be such an honest member. Because $n_m$ is honest, $n_m$ does not allow the fork. This contradicts the assumption. Thus, the lemma is proved.
\end{proof}

Each node $n_i$ has an OPERA chain $G_i$. We define a consistent chain from a sequence of OPERA chain $G_i$.
\begin{defn}[consistent chain] 
	A global consistent chain $G^C$ is a chain if $G^C \sim G_i$ for all $G_i$.
\end{defn}

We denote $G \sqsubseteq G'$ to stand for $G$ is a subgraph of $G'$.
\begin{lem}
	$\forall G_i$ ($G^C \sqsubseteq G_i$).
\end{lem}
\begin{lem}
	$\forall v \in G^C$ $\forall G_i$ ($G^C[v] \sqsubseteq G_i[v]$).
\end{lem}
\begin{lem}
	($\forall v_c \in G^C$) ($\forall v_p \in G_i$) (($v_p \hbefore v_c) \Rightarrow v_p \in G^C$).
\end{lem}

Now we state the following important propositions.
\begin{defn}[consistent root]
	Two chains $G_1$ and $G_2$ are root consistent, if for every $v$ contained in both chains, and $v$ is a root of $j$-th frame in $G_1$, then $v$ is a root of $j$-th frame in $G_2$.
\end{defn}
\begin{prop}
	If $G_1 \sim G_2$, then $G_1$ and $G_2$ are root consistent.
\end{prop}
\begin{proof}
	By consistent chains, if $G_1 \sim G_2$ and $v$ belongs to both chains, then $G_1[v]$ = $G_2[v]$.
	We can prove the proposition by induction. For $j$ = 0, the first root set is the same in both $G_1$ and $G_2$. Hence, it holds for $j$ = 0. Suppose that the proposition holds for every $j$ from 0 to $k$. We prove that it also holds for $j$= $k$ + 1.
	 Suppose that $v$ is a root of frame $f_{k+1}$ in $G_1$. 
	Then there exists a set $S$ reaching 2/3 of members in $G_1$ of frame $f_k$ such that $\forall u \in S$ ($u\hbefore v$). As $G_1 \sim G_2$, and $v$ in $G_2$, then $\forall u \in S$ ($u \in G_2$). Since the proposition holds for $j$=$k$, 
	As $u$ is a root of frame $f_{k}$ in $G_1$, $u$ is a root of frame $f_k$ in $G_2$. Hence, the set $S$ of 2/3 members $u$ happens before $v$ in $G_2$. So $v$ belongs to $f_{k+1}$ in $G_2$. The proposition is proved.
\end{proof}

From the above proposition, one can deduce the following:
\begin{lem}
		$G^C$ is root consistent with $G_i$ for all nodes.
 \end{lem}
Thus, all nodes have the same consistent root sets, which are the root sets in $G^C$. Frame numbers are consistent for all nodes.


\begin{defn}[Clotho]
	A root $r_k$ in the frame $f_{a+3}$ can nominate a root $r_a$ as Clotho if more than 2n/3 roots in the frame $f_{a+1}$ Happened-Before $r_a$ and $r_k$ Happened-Before the roots in the frame $f_{a+1}$.
\end{defn} 

\begin{lem}
	\label{lem:root}
	For any root set $R$, all nodes nominate same root into Clotho.
\end{lem}

\begin{proof}
	Based on Theorem~\ref{thm:sameopera}, each node nominates a root into Clotho via the flag table. If all nodes have an OPERA chain with same shape, the values in flag table should be equal to each other in OPERA chain. Thus, all nodes nominate the same root into Clotho since the OPERA chain of all nodes has same shape.
\end{proof}

\begin{lem}[Fork Lemma]
	\label{lem:fork}
	For any Clotho $v_c$ in an honest node, there does not exist a pair of forked event blocks ($v_i$,$v_j$) under $G[v_c]$.
\end{lem}
\begin{proof}
Assume there is a Clotho $v_c$ in node $n_i$ such that there exists a pair of forked events ($v_i$,$v_j$) under $G[v_c]$.

Recall that from previous Fork detection lemma, if an honest node contains any fork ($v_i$,$v_j$), then they cannot both be roots. One of the forked events was removed and the detected fork was notified across other nodes. Thus, $v_i$ and $v_j$ cannot both be roots. Let $r_i$ denote a root at lowest layer that knows both $v_i$ and $v_j$.
It is obvious that $r_i$ belongs to $G[v_c]$.
Since $v_c$ is a Clotho, let assume there is a root $r_j$ that nominates $v_c$.
There are two cases: (1) $r_i$ is different from both $v_i$ and $v_j$; (2) $r_i$ is one of $v_i$ and $v_j$. If $r_i$ is part of the fork, we can detect in $r_i$ and completely remove the fork in $r_j$.
\end{proof}

\begin{lem}
	For a fork that happened-before a root a root $v$ in OPERA chain, any root $u$ must have seen the fork before nominating $v$ to be a Clotho.
\end{lem}

\begin{proof}
	Suppose that a node creates forked event blocks ($v_x, v_y$). Suppose root $v$ happens before $v_x$ and $v_y$. Suppose there is a root $u$ that nominates $v$ as a Clotho. To create two Clothos that can reach both events, the event blocks should reach by more than 2n/3 nodes. Therefore, the OPERA chain can structurally detect the fork before roots become Clotho.
\end{proof}

\begin{thm}[OPERA consistency]
\label{thm:sameopera}
All nodes grows up into same consistent OPERA chain $G^C$, which is fork-free.
\end{thm}
\begin{proof}
Suppose that there are two event blocks $v_x$ and $v_y$ contained in both $G_1$ and $G_2$, and their path between $v_x$ and $v_y$ in $G_1$ is not equal to that in $G_2$. We can consider that the path difference between the nodes is a kind of fork attack. Based on Lemma~\ref{lem:fork}, if an attacker forks an event block, each chain of $G_1$ and $G_2$ can detect and remove the fork before the Clotho is generated. Thus, any two nodes have consistent OPERA chain. 
\end{proof}

\begin{thm}[H-OPERA consistency]
	\label{thm:samehopera}
	All nodes grows up into same consistent H-OPERA chain $H^C$.
\end{thm}
\begin{proof}
	The layering algorithm in $L_\phi$ when applies to the same graph $G$, produces the same layering information. Since all nodes builds into the same consistent OPERA chain $G^C$. The layering on $G^C$ is consistent across the nodes.
\end{proof}

\begin{prop}
	All nodes have the same Clothos and Atropos.
\end{prop}
\begin{thm}
	All nodes have consistent Main chains.
\end{thm}

Using the global consistent graph, 
all nodes grows into the same $G^C$. Thus, the set of Clothos and Atroposes is consistent in all the nodes.

\begin{thm}[Finality]
	\label{thm:finality}
	All nodes produce the same ordering and consensus time for finalized event blocks.
\end{thm}
\begin{proof}
Given the set of Clotho is consistent and the layering is consistent across the nodes, the topological sorting of Clotho based on layering is consistent. Thus, the ordering of vertices in each the subgraphs under a Clotho is the same across the nodes. Hence, the consensus time of final event blocks which is assigned based on the topological ordering, is consistent.
\end{proof}

\subsection{Semantics of Lachesis protocol}\label{sec:semantics} 

This section gives the formal semantics of Lachesis consensus protocol.
We use CCK model \cite{cck92} of an asynchronous system as the base of the semantics of our Lachesis protocol. Events are ordered based on Lamport's happens-before relation. 
In particular, we use Lamport’s theory to describe global states of an asynchronous system.

We present notations and concepts, which are important for Lachesis protocol. In several places, we adapt the notations and concepts of CCK paper to suit our Lachesis protocol. 

An asynchronous system consists of the following:
\begin{defn}[process]
	A process $p_i$ represents a machine or a node. The process identifier of $p_i$ is $i$. A set $P$ = \{1,...,$n$\} denotes the set of process identifiers.
\end{defn}
\begin{defn}[channel]
	A process $i$ can send messages to process $j$ if there is a channel ($i$,$j$). Let $C$ $\subseteq$ \{($i$,$j$) s.t. $i,j \in P$\} denote the set of channels.
\end{defn}
\begin{defn}[state]
	A local state of a process $i$ is denoted by $s_j^i$.
\end{defn}
A local state consists of a sequence of event blocks $s_j^i = v_0^i, v_1^i, \dots, v_j^i$. 

In a DAG-based protocol, each $v_j^i$ event block is valid only the reference blocks exist exist before it. From a local state $s_j^i$, one can reconstruct a unique DAG. That is, the mapping from a local state  $s_j^i$ into a DAG is \emph{injective} or one-to-one. 
Thus, for Lachesis, we can simply denote the $j$-th local state of a process $i$ by the OPERA chain $g_j^i$ (often we simply use $G_i$ to denote the current local state of a process $i$).

\begin{defn}[action]
	An action is a function from one local state to another local state.
\end{defn}
Generally speaking, an action can be either: a $send(m)$ action where $m$ is a message, a $receive(m)$ action, and an internal action. A message $m$ is a triple $\langle i,j,B \rangle$ where $i \in P$ is the sender of the message, $j \in P$ is the message recipient, and $B$ is the body of the message. Let $M$ denote the set of messages. 
In Lachesis protocol, $B$ consists of the content of an event block $v$. 
Semantics-wise, in Lachesis, there are  two actions that can change a process's local state: creating a new event and receiving an event from another process.

\begin{defn}[event] An event is a tuple $\langle  s,\alpha,s' \rangle$ consisting of a state, an action, and a state.
\end{defn}

Sometimes, the event can be represented by the end state $s'$.
The $j$-th event in history $h_i$ of process $i$ is $\langle  s_{j-1}^i,\alpha,s_j^i \rangle$, denoted by $v_j^i$.

\begin{defn}[local history] A local history $h_i$ of process $i$ is a (possibly infinite) sequence of alternating local states  --- beginning with a distinguished initial state. A set $H_i$ of possible local histories for each process $i$ in $P$.
\end{defn}

The state of a process can be obtained from its initial state and the sequence of actions or events that have occurred up to the current state. 
In Lachesis protocol, we use append-only sematics. The local history may be equivalently described as either of the following:
$$h_i = s_0^i,\alpha_1^i,\alpha_2^i, \alpha_3^i \dots $$
$$h_i = s_0^i, v_1^i,v_2^i, v_3^i \dots $$
$$h_i = s_0^i, s_1^i, s_2^i, s_3^i, \dots$$

In Lachesis, a local history is equivalently expressed as:
$$h_i = g_0^i, g_1^i, g_2^i, g_3^i, \dots$$
where $g_j^i$ is the $j$-th local OPERA chain (local state) of the process $i$.

\begin{defn}[run] Each asynchronous run is a vector of local histories. Denoted by
	$\sigma = \langle h_1,h_2,h_3,...h_N \rangle$.
\end{defn}

Let $\Sigma$ denote the set of asynchronous runs.

We can now use Lamport’s theory to talk about global states of an asynchronous system.
A global state of run $\sigma$ is an $n$-vector of prefixes of local histories of $\sigma$, one prefix per process.
The happens-before relation can be used to define a consistent global state, often termed a consistent cut, as follows.

\begin{defn}[Consistent cut] A consistent cut of a run $\sigma$ is any global state such that if $v_x^i \rightarrow v_y^j$ and $v_y^j$ is in the global state, then $v_x^i$ is also in the global state. Denoted by $\vec{c}(\sigma)$.
\end{defn}

By Theorem~\ref{thm:conchains}, all nodes have consistent local OPERA chains. The concept of consistent cut formalizes such a global state of a run. A consistent cut consists of all consistent OPERA chains. A received event block exists in the global state implies the existence of the original event block.
Note that a consistent cut is simply a vector of local states; we will use the notation $\vec{c}(\sigma)[i]$ to indicate the local state of $i$ in cut $\vec{c}$ of run $\sigma$.



The formal semantics of an asynchronous system is given via  the satisfaction relation $\vdash$. Intuitively $\vec{c}(\sigma) \vdash \phi$, ``$\vec{c}(\sigma)$ satisfies $\phi$,'' if fact $\phi$ is true in cut $\vec{c}$ of run $\sigma$. 
We assume that we are given a function $\pi$ that assigns a truth value to each primitive proposition $p$. The truth of a primitive proposition $p$ in $\vec{c}(\sigma)$ is determined by $\pi$ and $\vec{c}$. This defines $\vec{c}(\sigma) \vdash p$.

\begin{defn}[equivalent cuts]
	Two cuts $\vec{c}(\sigma)$ and $\vec{c'}(\sigma')$ are equivalent  with respect to $i$ if: $$\vec{c}(\sigma) \sim_i \vec{c'}(\sigma') \Leftrightarrow \vec{c}(\sigma)[i] = \vec{c'}(\sigma')[i]$$
\end{defn}

We introduce two families of modal operators, denoted by $K_i$ and $P_i$, respectively. Each family indexed by process identifiers. 
Given a fact $\phi$, the modal operators are defined as follows:
\begin{defn}[$i$ knows $\phi$]
	$K_i(\phi)$ represents the statement ``$\phi$ is true in all possible consistent global states that include $i$’s local state''. 
		$$\vec{c}(\sigma) \vdash K_i(\phi) \Leftrightarrow \forall \vec{c'}(\sigma')   (\vec{c'}(\sigma') \sim_i \vec{c}(\sigma) \ \Rightarrow\ \vec{c'}(\sigma') \vdash \phi) $$
\end{defn}

\begin{defn}[$i$ partially knows $\phi$]
	$P_i(\phi)$ represents the statement ``there is some consistent global state in this run that includes $i$’s local state, in which $\phi$ is true.''
		$$\vec{c}(\sigma) \vdash P_i(\phi) \Leftrightarrow \exists \vec{c'}(\sigma) ( \vec{c'}(\sigma) \sim_i \vec{c}(\sigma) \ \wedge\ \vec{c'}(\sigma) \vdash \phi )$$ 
\end{defn}

The next modal operator is written $M^C$ and stands for ``majority concurrently knows.'' This is adapted from the ``everyone concurrently knows'' in CCK paper~\cite{cck92}.
The definition of $M^C(\phi)$ is as follows.

\begin{defn}[majority concurrently knows]
	$$M^C(\phi) =_{def} \bigwedge_{i \in S} K_i P_i(\phi), $$ where $S \subseteq P$ and $|S| > 2n/3$.	
\end{defn}

In the presence of one-third of faulty nodes, the original operator ``everyone concurrently knows'' is sometimes not feasible.
Our modal operator $M^C(\phi)$ fits precisely the semantics for BFT systems, in which unreliable processes may exist.

The last modal operator is concurrent common knowledge (CCK), denoted by $C^C$.
\begin{defn}[concurrent common knowledge]
	$C^C(\phi)$ is defined as a fixed point of $M^C(\phi \wedge X)$
\end{defn}
CCK defines a state of process knowledge that implies that all processes are in that same state of knowledge, with respect to $\phi$, along some cut of the run. In other words, we want a state of knowledge $X$ satisfying: $X = M^C(\phi \wedge X)$.	
$C^C$ will be defined semantically as the weakest such fixed point, namely as the greatest fixed-point of $M^C(\phi \wedge X)$.
It therefore satisfies:
$$C^C(\phi) \Leftrightarrow  M^C(\phi \wedge C^C(\phi))$$

Thus, $P_i(\phi)$ states that there is some cut in the same asynchronous run $\sigma$ including $i$’s local state, such that $\phi$ is true in that cut.

Note that $\phi$ implies $P_i(\phi)$. But it is not the case, in general, that $P_i(\phi)$ implies $\phi$ or even that $M^C(\phi)$ implies $\phi$. The truth of $M^C(\phi)$ is determined with respect to some cut $\vec{c}(\sigma)$. A process cannot distinguish which cut, of the perhaps many cuts that are in the run and consistent with its local state, satisfies $\phi$; it can only know the existence of such a cut. 

\begin{defn}[global fact]
	Fact $\phi$ is valid in system $\Sigma$, denoted by $\Sigma \vdash \phi$, if $\phi$ is true in all cuts of all runs of $\Sigma$.	
	$$\Sigma \vdash \phi 
	\Leftrightarrow (\forall \sigma \in \Sigma)(\forall\vec{c}) (\vec{c}(a) \vdash \phi)$$
\end{defn} 

\begin{defn}
	Fact $\phi$ is valid, denoted $\vdash \phi$, if $\phi$ is valid in all systems, i.e. 
	$(\forall \Sigma) (\Sigma \vdash \phi)$.
\end{defn}

\begin{defn}[local fact]
	 A fact $\phi$ is local to process $i$ in system $\Sigma$ if
	 $\Sigma \vdash (\phi \Rightarrow K_i \phi)$
\end{defn}

\begin{thm} If $\phi$ is local to process $i$ in system $\Sigma$, then $\Sigma \vdash (P_i(\phi) \Rightarrow \phi)$.	
\end{thm}

\begin{lem}
	If fact $\phi$ is local to 2/3 of the processes in a system $\Sigma$, then $\Sigma \vdash (M^C(\phi) \Rightarrow \phi)$ and furthermore $\Sigma \vdash (C^C(\phi) \Rightarrow \phi)$.
	\end{lem}

\begin{defn}
	A fact $\phi$ is attained in run $\sigma$ if $\exists \vec{c}(\sigma) (\vec{c}(\sigma) \vdash \phi)$.
\end{defn}

Often, we refer to ``knowing'' a fact $\phi$ in a state rather than in a consistent cut, since knowledge is dependent only on the local state of a process.
Formally, $i$ knows $\phi$ in state $s$ is shorthand for
$$\forall \vec{c}(\sigma) (\vec{c}(\sigma)[i] = s \Rightarrow \vec{c}(\sigma) \vdash \phi)$$

For example, if a process in Lachesis protocol knows a fork exists (i.e., $\phi$ is the exsistenc of fork) in its local state $s$ (i.e., $g_j^i$), then a consistent cut contains the state $s$ will know the existence of that fork.

\begin{defn}[$i$ learns $\phi$]
 Process $i$ learns $\phi$ in state $s_j^i$ of run $\sigma$ if $i$ knows $\phi$ in $s_j^i$ and, for all previous states $s_k^i$ in run $\sigma$, $k < j$, $i$ does not know $\phi$.
 \end{defn}

The following theorem says that if $C_C(\phi$ is attained in a run then all processes $i$ learn $P_i C^C(\phi)$ along a single consistent cut.

\begin{thm}[attainment]
	 If $C^C(\phi)$ is attained in a run $\sigma$, then the set of states in which all processes learn $P_i C^C(\phi)$ forms a consistent cut in $\sigma$.	
\end{thm}

We have presented a formal semantics of Lachesis protocol based on the concepts and notations of concurrent common knowledge~\cite{cck92}.
For a proof of the above theorems and lemmas in this Section, we can use similar proofs as described in the original CCK paper.

With the formal semantics of Lachesis, the theorems and lemmas described in Section~\ref{se:proof} can be expressed in term of CCK. For example, one can study a fact $\phi$ (or some primitive proposition $p$) in the following forms: `is there any existence of fork?'. One can make Lachesis-specific questions like 'is event block $v$ a root?', 'is $v$ a clotho?', or 'is $v$ a atropos?'. This is a remarkable result, since we are the first that define such a formal semantics for DAG-based protocol.




\clearpage
\section{Reference}\label{se:ref}

\renewcommand\refname{\vskip -1cm}
\bibliographystyle{unsrt}
\bibliography{LCA}

\end{document}